\documentclass{IEEEtran}
\usepackage{graphicx}
\usepackage{amsmath, amssymb, bm,bbm}
\usepackage{amsthm,cases}
\usepackage{subfigure}
\usepackage{epstopdf}
\usepackage{cite}
\usepackage{latexsym}
\usepackage{color}
\usepackage{epstopdf}
\usepackage{arydshln}
\usepackage{amsfonts}
\usepackage{accents}
\usepackage{multirow}
\usepackage{dsfont}
\usepackage{enumitem}
\newtheorem{proposition}{Proposition}

\newtheorem{theorem}{Theorem}

  \makeatletter
  \def\widebar{\accentset{{\cc@style\underline{\mskip10mu}}}}
  \makeatother

\newcommand{\la}{\lambda}
\newcommand{\li}{\lambda_i}
\newcommand{\argmax}{\mathop{\rm arg~max}\limits}
\newcommand{\argmin}{\mathop{\rm arg~min}\limits}

\begin{document}

\title{Eigendecomposition-Free Sampling Set Selection for Graph Signals}%

\author{Akie~Sakiyama, Yuichi~Tanaka, Toshihisa~Tanaka, and~Antonio~Ortega
\thanks{This work was supported in part by JST PRESTO under Grant JPMJPR1656.}
\thanks{A. Sakiyama and Y. Tanaka are with the Graduate School of BASE, Tokyo University of Agriculture and Technology, Koganei, Tokyo, 184-8588, Japan. Y. Tanaka is also with PRESTO, Japan Science and Technology Agency, Kawaguchi, Saitama, 332-0012, Japan (email: sakiyama@msp-lab.org; ytnk@cc.tuat.ac.jp).}
\thanks{T. Tanaka is with the Department of Electronic and Information Engineering, Tokyo University of Agriculture and Technology, Koganei, Tokyo, 184-8588, Japan (e-mail: tanakat@cc.tuat.ac.jp)}
\thanks{A. Ortega is with the University of Southern California, Los Angeles, CA 90089 USA (email: antonio.ortega@sipi.usc.edu).}
}

\maketitle

\begin{abstract}
This paper addresses the problem of selecting an optimal sampling set for signals on graphs. The proposed sampling set selection (SSS) is based on a localization operator that can consider both vertex domain and spectral domain localization. We clarify the relationships among the proposed method, sensor position selection methods in machine learning, and conventional SSS methods based on graph frequency. In contrast to the conventional graph signal processing-based approaches, the proposed method does not need to compute the eigendecomposition of a variation operator, while still considering (graph) frequency information. We evaluate the performance of our approach through comparisons of prediction errors and execution time.
\end{abstract}
%

\section{Introduction}
\subsection{Motivation}
Graphs give intuitive and effective representations for visualizing or investigating large quantities of intricately interrelated data. Network topologies have been studied in graph theory for a long time. In the past half-decade, the theory of analyzing and processing data on the vertices of a graph as well as underlying graph topologies, namely, signal processing on graphs, has been developed rapidly\cite{Shuman2013, Sandry2013, Sandry2014,Ortega2018}.
This theory enables us to efficiently apply signal processing techniques to many practical problems, such as social \cite{Sandry2013}, traffic \cite{Crovel2003}, brain  \cite{Higash2016, Huang2018}, and sensor networks \cite{Shen2010, Leonar2013}, following approaches similar to those used for audio, image, or time domain signals in traditional signal processing. 
This paper considers sampling methods for graph signals, a key topic in the development of signal processing on graphs.
 
The sampling of graph signals is an essential task for treating big or complex-structured data in the real world \cite{Shuman2013, Sandry2013}. Handling such raw data consumes a significant amount of system resources, both storage and computation, and sampled versions capturing most of the relevant information in the data would thus be highly desirable.

A main challenge in graph signal sampling is that in general there is no such thing as ``regular sampling,'' and thus the sampling set has to be optimized based on the topology of the underlying graph. Many different approaches have been proposed for sampling set selection (SSS) on graphs \cite{Anis2014, Narang2013a,Gadde2014,Gadde2015,Chen2015,Shomor2014,Anis2016,Marque2016,Tanaka2018}.
Unfortunately, most of these methods have high computational complexity, as they require eigendecompositions to compute the graph Fourier basis (or some of its vectors).

SSS can be classified into deterministic and random sampling methods. Deterministic approaches \cite{Anis2014, Narang2013a,Gadde2014,Gadde2015,Chen2015,Shomor2014,Anis2016,Marque2016} select vertices (often one-by-one) such that a target cost function is minimized or maximized by each selection, whereas random methods \cite{Puy2018, Perrau2018} select vertices randomly according to some pre-computed probability distribution. In this study, we focus on the deterministic approach because it has the following advantage with respect to random sampling.
    
In random sampling-based methods, nothing prevents selecting vertices that have similar importance (higher probability) but that happen to be close to each other. We observe that this often happens, especially when graphs have irregular degree distributions and therefore the distribution of probabilities is also biased. When a high-probability vertex is selected and then a nearby vertex is also chosen, the second vertex may not lead to improvements in reconstruction performance. Thus, in practice, random sampling methods may perform well on average, but often require more samples than deterministic methods to achieve the same reconstruction quality.

Deterministic SSS techniques are based on selecting vertices for minimizing the reconstruction error when signals are reconstructed from their samples. They have been studied in the context of sampling theorems for graph signals \cite{Chen2015, Tsitsv2016, Anis2016}. They define a cost function based on the assumption that the reconstruction is performed using ideal filters under different optimality criteria (e.g., average case or worst case). 

Recently, a vertex-localized SSS was proposed \cite{Jayawa2018}. This is a two-step algorithm, whereby vertices are first screened to obtain a \textit{permissible set} of vertices, i.e., vertices that are far enough from those vertices that have already been selected. Then, an optimal vertex is selected from the permissible set. Although this is conceptually similar to our approach, ours is a one-step algorithm. To select sufficiently far vertices, we control the vertex/spectral spread using graph spectral filters other than the ideal filters.

The above SSS methods are summarized in Table \ref{tab:sss_summary}. The abbreviations of the deterministic methods are found in Section \ref{subsec:sss_gft}, along with the cost functions used in each case.

\begin{table*}[]
    \centering
    \caption{Comparison of Graph-based Sampling Set Selection Methods}
    \label{tab:sss_summary}
    \begin{tabular}{l|c|c|c|c}
    \hline
    & Deterministic/ & Kernel & Localization & Localization \\
    & random & & in vertex domain & in graph freq. domain\\\hline
    Cumulative coherence \cite{Puy2018} & Random & Ideal & \checkmark$^*$ & \checkmark\\
    Global/local uncertainty\cite{Perrau2018} & Random & Arbitrary &\checkmark & \checkmark\\
    MaxCutoff \cite{Anis2016} &  Deterministic & $\lambda^k$ ($k \in \mathbb{Z}_+$) & & \checkmark\\
    MinSpec/MinTrac \cite{Chen2015} &Deterministic & Ideal & & \checkmark\\
    MinFrob/MaxFrob/MaxPVol \cite{Tsitsv2016} &Deterministic & Ideal & & \checkmark\\
    Vertex screening \cite{Jayawa2018} & Deterministic & Ideal & \checkmark & \checkmark\\\hline
    Proposed method & Deterministic & Arbitrary & \checkmark & \checkmark \\
    \hline
    \multicolumn{5}{l}{$^*$ Localized in the vertex domain only if the ideal kernel is approximated by a polynomial}\\
    \end{tabular}
\end{table*}

In this paper, we propose a deterministic sampling method for graph signals based on the graph localization operator \cite{Perrau2018} and reveal the relationship among the sensor selection methods based on the Gaussian process \cite{Cressi2015, Shewry1987,Krause2008,Sharma2015}, the conventional graph sampling methods \cite{Anis2016, Chen2015, Tsitsv2016}, and the proposed method. The localization operator is introduced in the context of the uncertainty principle of graph signals \cite{Perrau2018}. It is the vertex domain operator with consideration of the graph frequency domain information. 

Our contributions in this paper are summarized as follows.
\begin{itemize}
\item Using the localized operator for SSS, the following benefits are obtained: a) graph frequency localization makes it possible to mimic the frequency-based SSS criteria of \cite{Anis2016, Chen2015, Shomor2014}, b) vertex localization is used to enable distributed SSS, and c) polynomial localization operators lead to lower complexity, i.e., eigendecomposition-free algorithms (see also Table \ref{tab:sss_summary}). This makes the SSS algorithm significantly faster (see Section \ref{sec:exp}).
\item We provide a unifying framework for many SSS techniques proposed to date as special cases of the localization operator-based SSS with different kernels and different optimization criteria to minimize the error. Even methods that were not initially viewed from a graph perspective, e.g., methods based on entropy \cite{Cressi2015, Shewry1987} and mutual information \cite{Krause2008,Sharma2015}, are included as its special cases (see Section \ref{sec:rela}).
\end{itemize}

Our preliminary work \cite{Sakiya2016,Sakiya2017} partially solved the problem of sensor position selection of sensor networks \cite{Krause2008, Deshpa2004, Krause2007, Uddin2014} using sampling theory for graph signals \cite{Anis2016, Chen2015, Shomor2014} and proposed a sensor selection method based on the localization operator. This paper adds many theoretical and practical implications. Specifically, we newly propose sampling approaches based on error minimization and clarify the relationships between the conventional sampling methods for graph signals and the proposed methods.

In the experiment, we present the execution time and prediction error comparisons to evaluate the performance of the proposed approach. The proposed method is faster and shows better performance than the conventional approaches\cite{Anis2016, Chen2015, Tsitsv2016}.

The rest of this paper is organized as follows. The preliminaries and notation are summarized in Section \ref{sec:pre}. Section \ref{sec:conv} introduces the sensor position selection approaches and the conventional graph sampling methods based on the graph Fourier basis. Section \ref{sec:pro} provides the signal reconstruction method and describes the proposed vertex and signal selection algorithm based on the graph localization operator. The section also compares the computational complexities of the conventional and proposed methods. Section \ref{sec:rela} clarifies that the proposed method has a deep connection with the conventional approaches introduced in Section \ref{sec:conv}. Section \ref{sec:exp} shows the experimental results of SSS and predicts the signals on unobserved vertices. Finally, Section \ref{sec:con} concludes the paper.

\subsection{Preliminaries and Notation}
\label{sec:pre}
A graph is represented as $\mathcal{G} = (\mathcal{V}, \mathcal{E})$, where $\mathcal{V}$ and $\mathcal{E}$ denote sets of vertices and edges, respectively. A graph signal is defined as $\bm{f}\in \mathbb{R}^N$, where $N$ is the number of vertices. We will only consider a connected, finite, undirected graph with no multiple edges. The variation operators are used for frequency analysis of graph signals. Although this paper mainly uses the graph Laplacian, we can use any variation operators, such as the adjacency matrix. 

The combinatorial graph Laplacian is defined as $\mathbf{ L}:=\mathbf{ D}-\mathbf{ A}$, where $\mathbf{ A}$ is the adjacency matrix whose $(m,n)$th element is the weight of the edge between $m$ and $n$ if $m$ and $n$ are connected, and 0 otherwise, and a diagonal matrix $\mathbf{ D}$ is the degree matrix whose $m$th diagonal element is $D({m,m}) = \sum_n A({m,n})$. 
The $i$th eigenvalue of $\mathbf{ L}$ is ${\li}$, which can be ordered, without loss of generality as: $0=\la_0<\la_1\le\la_2\ldots \leq\la_{N-1}=\lambda_{\text{max}}$, and its eigenvector is $\bm{u}_{i}\in\mathbb{C}^N$. 

The graph Fourier transform is defined as follows \cite{Chung1997,Hammon2011}:
$\widebar{\bm{f}} = \mathbf{ U}^*\bm{f},$
where $\mathbf{ U}=[\bm{u}_0\ldots \bm{u}_{N-1}]$ and $\cdot^*$ is the conjugate transpose of a matrix or a vector.
The inverse graph Fourier transform is $\bm{f} = \mathbf{U}\widebar{\bm{f}}$. 
Let $h(\li)$ be the spectral kernel; then, the filtering in the graph frequency domain can be written as $\bm{f}_\text{out}=\mathbf{ U}h(\mathbf{ \Lambda})\mathbf{ U}^*\bm{f}_\text{in}$, where $\mathbf{ \Lambda}=\text{diag}(\lambda_0, \ldots, \lambda_{N-1})$.

The $n$th element of the localization operator on the center vertex $i$ is defined as \cite{Perrau2018}
\begin{equation}
{\displaystyle T_{g,i}(n)=\sqrt{N}\sum^{N-1}_{l=0}{g}(\lambda_l)u_l^*(i)u_l(n),}
\label{eq:estcf}
\end{equation}
where ${g}(\lambda)$ is an arbitrary filter kernel.
The matrix arranging localization operator in a row is
\begin{equation}
\mathbf{T}=[\bm{T}_{g,0}\ \bm{T}_{g,1}\cdots\bm{T}_{g, N-1}]=\mathbf{ U}g(\mathbf{ \Lambda})\mathbf{ U}^*.
\label{loc}
\end{equation}

The other notation used in this paper is summarized in Table \ref{not}.

\begin{table}[t]
  \centering 
  \caption{Notation Used in This Paper: $\bm{x}\in \mathbb{R}^M$, $\mathbf{ X}\in \mathbb{R}^{M\times L}$ And $\mathbf{ Y}\in \mathbb{R}^{M\times M}$, And $\mathcal{A}$ And $\mathcal{B}$ Are Arbitrary Vector, Matrices And Sets, Respectively.}\label{t1}
  \begin{tabular}{c|l}
\hline
Symbol& Description\\\hline
$\det[\mathbf{ X}]$& determinant of $\mathbf{X}$\\[2pt]
$\text{tr}[\mathbf{ X}]$& trace of $\mathbf{X}$\\[2pt]
$|\mathbf{ X}|$&$\text{sgn}(\mathbf{ X})\circ\mathbf{ X}$\\[2pt]
$|\bm{x}|$&$\text{diag}(\text{sgn}(\bm{x}))\bm{x}$\\[2pt]
$|{\mathcal X}|$&number of elements in ${\mathcal X}$\\[2pt]
$\bm{x}_\mathcal{A}$&restriction of $\bm{x}$ to its components indexed by $\mathcal{A}$\\[2pt]
$\mathbf{ X}_{\mathcal{A}\mathcal{B}}$& restriction of $\mathbf{ X}$ to its rows by $\mathcal{A}$ and columns by $\mathcal{B}$\\[2pt]
 $\mathbf{ X}_{\mathcal{A}}$&$\mathbf{ X}_{\mathcal{A}\mathcal{A}}$\\[2pt]
\multirow{2}{*}{$\mu_i(\mathbf{ Y})$}& $i$th eigenvalue of $\mathbf{Y}$ \\
&$\mu_\text{min}(\mathbf{ Y})=\mu_0(\mathbf{ Y})\leq\cdots\leq\mu_{N-1}(\mathbf{ Y})=\mu_\text{max}(\mathbf{ Y})$\\[2pt]
 $\bm{v}_i(\mathbf{ Y})$& eigenvector of $\mathbf{ Y}$ corresponding to $\mu_i(\mathbf{ Y})$\\[2pt]
\multirow{2}{*}{$\sigma_i(\mathbf{ X})$}&$i$th singular value of $\mathbf{ X}$\\
&$\sigma_\text{min}(\mathbf{ X})=\sigma_0(\mathbf{ X})\leq\cdots\leq\sigma_{N-1}(\mathbf{ X})$\\[2pt]
  $\mathbf{ 1}_{\mathcal{A}}$ & ${1}_{\mathcal{A}}(m)=1$ if $m\in\mathcal{A}$ and $0$ otherwise\\
  \hline
  \end{tabular}
  \label{not}
  \end{table}

\section{Conventional Approaches for Sampling Set Selection}
\label{sec:conv}
We briefly introduce the objective functions of the conventional methods for selecting sensor locations and sampling points of graph signals. Their derivations are described in Appendices A and B.
We consider the problem of selecting $|\mathcal{S}|=F$ points, where $\mathcal{S}$ is the set of selected locations (for sensor selections) or vertices (for graph sampling theories), out of $|\mathcal{V}|=N$ possible locations or vertices in the original graph.

\subsection{Sensor Position Selection Based on Gaussian Process}
Sensor position selection algorithms have been developed in the area of machine learning. One of the major methods assumes that the spatial phenomena are modeled as a Gaussian process (GP) and, therefore, the stochastic signal $\bm{f}$ has the following Gaussian joint zero-mean distribution\cite{Deshpa2004}:
\begin{equation}
p(\bm{f})= \frac{1}{(2\pi)^{\frac{N}{2}}\det[\mathbf{ K}]}\exp{\left(-\frac{1}{2}\bm{f}^T\mathbf{ K}^{-1}\bm{f}\right)},
\label{eq:dis1}
\end{equation}
where $\cdot^T$ is the transpose of a matrix or a vector, and $\mathbf{ K}\in \mathbb{R}^{N\times N}$ is the covariance matrix of all locations $\mathcal{V}$ whose $(i, j)$th element is $\mathcal{K}(i,j)$ with a symmetric positive-definite kernel function $\mathcal{K}(\cdot,\cdot)$. The benefit of the GP model is that, if the signal $\bm{f}$ is distributed according to a multivariate Gaussian, the marginal and conditional distributions of its subset signal ${f}({y})$, where $y\in \mathcal{V}$, are also Gaussian with conditional variance $\sigma^2_{y|\mathcal{S}}=\mathcal{K}(y,y)-\mathbf{ K}_{y\mathcal{S}}\mathbf{ K}_{\mathcal{S}}^{-1}\mathbf{ K}_{\mathcal{S}y}$. Under this assumption, sensors are placed at the most informative locations.

\subsubsection{Entropy\cite{Cressi2015, Shewry1987}}
The objective function is
\begin{equation}
\mathcal{S}^*=\argmax_{\mathcal{S}\subset\mathcal{V}:|\mathcal{S}|=F}\log\det{[\mathbf{ K}_{\mathcal{S}}]}.
\label{eq:en}
\end{equation}
A greedy algorithm that adds sensor $y^*$ satisfying following condition to $\mathcal{S}$ one by one is used for optimization:
\begin{equation}
y^*\leftarrow\argmax_{y\in\mathcal{S}^c_m}{\mathcal{K}(y,y)-\mathbf{ K}_{y\mathcal{S}_m}\mathbf{ K}_{\mathcal{S}_m}^{-1}\mathbf{ K}_{\mathcal{S}_my}},
\label{eq:en3}
\end{equation}
where ${\mathcal{S}_m}$ are the already selected vertices in the $m$th iteration, ${\mathcal{S}_m^c}=\mathcal{V}\setminus\mathcal{S}_m$. 
\subsubsection{Mutual Information (MI) \cite{Krause2008,Sharma2015}} The objective function is
\begin{equation}
\mathcal{S}^*=\argmax_{\mathcal{S}\subset\mathcal{V}:|\mathcal{S}|=F}\log\det{[\mathbf{ K}_{\mathcal{S}}]}+\log\det{[\mathbf{ K}_{\mathcal{S}^c}]},
\label{eq:mi1}
\end{equation}
where ${\mathcal{S}^c}=\mathcal{V}\setminus\mathcal{S}$. A greedy algorithm is also used for optimization:
\begin{equation}
y^*\leftarrow\argmax_{y\in{\mathcal{S}}_m^c}\ \frac{\mathcal{K}(y,y)-\mathbf{ K}_{y{\mathcal{S}}_m}\mathbf{ K}_{{\mathcal{S}}_m}^{-1}\mathbf{ K}_{{\mathcal{S}}_my}}{\mathcal{K}(y,y)-\mathbf{ K}_{y\widebar{{\mathcal{S}}_m}}\mathbf{ K}_{\widebar{{\mathcal{S}}_m}}^{-1}\mathbf{ K}_{\widebar{{\mathcal{S}}_m}y}},
\label{eq:mi3}
\end{equation}
where $\widebar{{\mathcal{S}}}_m={\mathcal{V}}\setminus({\mathcal{S}_m}\cup y)$.

\subsection{Graph Sampling Based on Graph Fourier Basis}\label{subsec:sss_gft}
Sampling methods based on graph frequency consider the problem of reconstructing bandlimited graph signals from their subsampled versions \cite{Gadde2014,Gadde2015,Chen2015,Shomor2014}.
Note that we do not need to assume the GP model for the graph signal processing-based approaches (including the proposed approach).

Let us define $\omega$- (for \cite{Anis2016}) and $|\mathcal{F}|$- (for \cite{Chen2015,Tsitsv2016}) bandlimited graph signals as the signals that have zero graph Fourier coefficients corresponding to the eigenvalues greater than $\omega$ and $\lambda_{|\mathcal{F}|-1}$, respectively: $\widebar{f}(i)=0$ for $\li>\omega$ or $i\ge|\mathcal{F}|$, where $\mathcal{F}$ is the set of indices associated with nonzero graph Fourier coefficients.

\subsubsection{Based on Cutoff Frequency (MaxCutoff)}
The objective function\cite{Anis2016} is
\begin{equation}
{\mathcal{S}}^*=\argmax_{\mathcal{S}\subset\mathcal{V}:|\mathcal{S}|=F}\mu_\text{min}((\mathbf{ L}^k)_{\mathcal{S}^c}).
\label{gs1}
\end{equation}
The objective for a greedy optimization is represented as 
\begin{equation}
y^*\leftarrow\argmax_{y\in\mathcal{S}^c_m}\ [v_\text{min}^2(({\mathbf{ L}}^k)_{\mathcal{S}^c_m})](y).
\end{equation}

The signal recovered from the sampled one is calculated as
\begin{equation}
\widehat{\bm{f}}=\mathbf{ U}_{\mathcal{V}\mathcal{F}}\mathbf{ U}_{\mathcal{S}\mathcal{F}}^+\bm{f}_{\mathcal{S}},
\label{rec}
\end{equation}
where $\mathcal{F}$ is the set of eigenvalues less than or equal to the estimated cutoff frequency $\Omega_k(\mathcal{S})$ and $\cdot^+$ represents the pseudoinverse of a matrix.
If the original signal $\bm{f}$ is $\Omega_k(\mathcal{S})$-bandlimited, it can be perfectly recovered using \eqref{rec}.

\subsubsection{Based on Error Minimization}
\cite{Chen2015} assumes that $|\mathcal{S}|\geq|\mathcal{F}|$ and uses \eqref{rec} for the reconstruction. It proposes two objective functions for selecting optimal sampling sets:
\begin{itemize}
\item\textit{MinSpec:}
\begin{equation}
{\mathcal{S}}^*=\argmax_{\mathcal{S}\subset\mathcal{V}:|\mathcal{S}|=F} \sigma_{\text{min}}(\mathbf{ U}_{\mathcal{S}\mathcal{F}}).
\label{gs2-1}
\end{equation}
The greedy algorithm is used to optimize this problem:
\begin{equation}
y^*\leftarrow\argmax_{y\in \mathcal{S}^c_m}\sigma_\text{min}(\mathbf{ U}_{\mathcal{F}({\mathcal{S}_m}\cup y)}).
\end{equation}
\item\textit{MinTrac:}
\begin{equation}
\mathcal{S}^*=\argmin_{\mathcal{S}\subset\mathcal{V}:|\mathcal{S}|=F}\text{tr}[(\mathbf{ U}_{\mathcal{S}\mathcal{F}}^*\mathbf{ U}_{\mathcal{S}\mathcal{F}})^{-1}].
\label{gs2-2}
\end{equation}
This also uses a greedy algorithm, which selects the vertex $y^*$:
\begin{equation}
y^*\leftarrow\argmin_{y\in \mathcal{S}^c_m}\text{tr}[(\mathbf{ U}_{\mathcal{F}({\mathcal{S}_m}\cup y)}^*\mathbf{ U}_{\mathcal{F}({\mathcal{S}_m}\cup y)})^{-1}].
\end{equation}
\end{itemize}

\subsubsection{Based on Localized Basis}
\cite{Tsitsv2016} also assumes that $|\mathcal{S}|\geq|\mathcal{F}|$ and uses the following interpolation for recovering the original signal $\bm{f}$ from the sampled signal $\mathbf{ D}_\text{ver}\bm{f}$:
\begin{equation}
\begin{split}
\widehat{\bm{f}}=(\mathbf{ D}_\text{sp}\mathbf{ D}_\text{ver}\mathbf{ D}_\text{sp})^+\mathbf{ D}_\text{ver}\bm{f},
\end{split}
\label{rec_tes}
\end{equation}
where $\mathbf{ D}_\text{ver}=\text{diag}(\mathbf{ 1}_{\mathcal{S}})$ and $\mathbf{ D}_\text{sp}=\mathbf{ U}\text{diag}(\mathbf{ 1}_{\mathcal{F}})\mathbf{ U}^*$ are the sampling operator in the vertex domain and the bandlimiting operator in the graph frequency domain, respectively.
This approach can perfectly recover the original signal $\bm{f}$ if it is $|\mathcal{F}|$-bandlimited. 

There are three objective functions: 
\begin{itemize}
\item\textit{MinFrob:} 
\begin{equation}
\mathcal{S}^*=\argmin_{\mathcal{S}\subset\mathcal{V}:|\mathcal{S}|=F} \|(\text{diag}(\bm{1}_{\mathcal{F}})\mathbf{ U}^*\mathbf{ D}_\text{ver})^+\|_F.\\
\label{gs3-1}
\end{equation}
For this metric a greedy algorithm for optimization selects a vertex $y^*$ at the $m$th step:
\begin{equation}
y^*\leftarrow\argmin_{y\in\mathcal{S}_m^c}\sum_{i=0}^{m-1}\frac{1}{\sigma_i(\mathbf{ U}_{({\mathcal{S}_m}\cup y)\mathcal{F}}^*)}.
\end{equation}

\item\textit{MaxFrob:} 
\begin{equation}
\mathcal{S}^*=\argmax_{\mathcal{S}\subset\mathcal{V}:|\mathcal{S}|=F} \|\mathbf{ D}_\text{sp}\mathbf{ D}_\text{ver}\mathbf{ D}_\text{sp}\|_F.
\label{gs3-2}
\end{equation}
This can be solved by a simple strategy that selects the $|\mathcal{S}|$ columns of $\mathbf{ U}_{{\mathcal{S}}\mathcal{F}}^*$ that have the maximum $\ell_2$ norm.

\item\textit{MaxPVol:} 
\begin{equation}
\mathcal{S}^*=\argmax_{\mathcal{S}\subset\mathcal{V}:|\mathcal{S}|=F} \det[\mathbf{ U}_{\mathcal{SF}}\mathbf{ U}_{\mathcal{SF}}^*].
\label{gs3-3}
\end{equation}
This method also uses a greedy algorithm and the sampled vertex at the $m$th iteration is selected as
\begin{equation}
y^*\leftarrow\argmax_{y\in\mathcal{S}^c_m}\prod_{i=0}^{m-1}\mu_i(\mathbf{ U}_{{(\mathcal{S}_m\cup y)}\mathcal{F}}\mathbf{ U}_{{(\mathcal{S}_m\cup y)}\mathcal{F}}^*).
\end{equation}
\end{itemize}

\begin{table*}[t]
  \centering 
  \caption{Sampling methods with localization operator: $\mathbf{ T}^\text{L}$, $\mathbf{ T}^{\text{K}}$ and $\mathbf{ T}^\text{I}$ are  $\sqrt{N}\mathbf{ U}g(\mathbf{ \Lambda})\mathbf{ U}^*=[\bm{T}_0\bm{g}\ \bm{T}_1\bm{g}\ldots\ \bm{T}_{N-1}\bm{g}]$ with kernel $g(\lambda)=\lambda$, $g(\lambda)=\lambda^{-1}+\delta$ and ideal kernel $g(\lambda)=1$ if $\lambda\in \mathcal{F}$ and $0$ otherwise, respectively.}
  \begin{tabular}{c|c|l|l}
\hline
\multicolumn{2}{c|}{} & \multicolumn{1}{c|}{Objective}&\multicolumn{1}{c}{Objective w/ Localized Operator}\\\hline\hline
\multirow{3}{*}{GP-based approach}&Entropy\cite{Krause2008}& $\argmax_{\mathcal{S}\subset\mathcal{V}:|\mathcal{S}|=F}\log\det{[\mathbf{ K}_{\mathcal{S}}]}$& $\argmax_{\mathcal{S}\subset\mathcal{V}:|\mathcal{S}|=F}\det{[\bf T^\text{K}_{\mathcal{S}}]}$\\
&MI\cite{Krause2008} & $\argmax_{\mathcal{S}\subset\mathcal{V}:|\mathcal{S}|=F}\log\det{[\mathbf{ K}_{\mathcal{S}}]}+\log\det{[\mathbf{ K}_{{\mathcal{S}^c}}]}$& $\argmax_{\mathcal{S}\subset\mathcal{V}:|\mathcal{S}|=F}\det{[\mathbf{ T}^\text{K}_{\mathcal{S}}]}\det{[\mathbf{ T}^\text{K}_{{\mathcal{S}^c}}]}$\\\hline
\multirow{6}{*}{Graph frequency-based approach}
& MaxCutoff \cite{Anis2016}&$\argmax_{\mathcal{S}\subset\mathcal{V}:|\mathcal{S}|=F}\lambda_\text{min}((\mathbf{ L}^k)_{{\mathcal{S}^c}})$&$\argmin_{\mathcal{S}\subset\mathcal{V}:|\mathcal{S}|=F}\|(((\mathbf{ T}^{\text{L}})^{k})_{{\mathcal{S}^c}})^{-1}\|_2$\\
&MinSpec\cite{Chen2015}& $\argmin_{\mathcal{S}\subset\mathcal{V}:|\mathcal{S}|=F}\|\mathbf{ U}_{\mathcal{SF}}^+\|_2$& $\argmin_{\mathcal{S}\subset\mathcal{V}:|\mathcal{S}|=F}\|(\mathbf{ T}^\text{I}_{\mathcal{SV}})^+\|_2$\\
&MinTrac\cite{Chen2015}& $\argmin_{\mathcal{S}\subset\mathcal{V}:|\mathcal{S}|=F}\text{tr}[(\mathbf{ U}_{\mathcal{S}\mathcal{F}}^*\mathbf{ U}_{\mathcal{S}\mathcal{F}})^{-1}]$& $\argmin_{\mathcal{S}\subset\mathcal{V}:|\mathcal{S}|=F}\text{tr}[(\mathbf{ T}^\text{I}_\mathcal{S})^{-1}]$\\
&MinFrob \cite{Tsitsv2016}& $\argmin_{\mathcal{S}\subset\mathcal{V}:|\mathcal{S}|=F}\|(\mathbf{ D}_\text{sp}\mathbf{ D}_\text{ver}\mathbf{ D}_\text{sp})^+\|_F$& $\argmin_{\mathcal{S}\subset\mathcal{V}:|\mathcal{S}|=F}\text{tr}[(\mathbf{ T}^\text{I}_\mathcal{S})^{-1}]$\\
&MaxFrob\cite{Tsitsv2016}& $\argmax_{\mathcal{S}\subset\mathcal{V}:|\mathcal{S}|=F}\|\mathbf{ D}_\text{sp}\mathbf{ D}_\text{ver}\mathbf{ D}_\text{sp}\|_F$&$\argmax_{\mathcal{S}\subset\mathcal{V}:|\mathcal{S}|=F}\text{tr}[(\mathbf{ T}^\text{I}_\mathcal{S})]$\\
&MaxPVol\cite{Tsitsv2016}&$\argmax_{\mathcal{S}\subset\mathcal{V}:|\mathcal{S}|=F} \det{[\mathbf{ U}_{\mathcal{SF}}^*\mathbf{ U}_{\mathcal{SF}}]}$&$\argmax_{\mathcal{S}\subset\mathcal{V}:|\mathcal{S}|=F} \det{[\mathbf{ T}^\text{I}_{\mathcal{S}}]}$\\
\hline
  \end{tabular}
  \label{tab:con}
  \end{table*}
The conventional methods and their objective functions are summarized in Table \ref{tab:con}.

\section{Vertex Selection Based on Localization Operator}
\label{sec:pro}
This section introduces the proposed SSS. 
First, we present the reconstruction method of missing graph signals based on the localization operator in \eqref{loc}. The sampled vertices are selected to minimize the reconstruction error or maximize the information corresponding to the localization operator. The computational complexities of the proposed and conventional methods are also discussed in this section.

\subsection{Reconstruction Method}
In our method, the missing values are reconstructed by a linear combination of
$\widetilde{\bm{T}}_{g,j}:=(\mathbf{ T}^k)_{{j}{\mathcal V}}=(\mathbf{ U}g(\mathbf{ \Lambda})^k\mathbf{ U}^*)_{{j}{\mathcal V}}$
with arbitrary kernel $g(\cdot)$, i.e., the sampled signal $\bm{f}_{\mathcal S}$ is recovered as follows:
\begin{equation}
\widehat{\bm{f}}_k=\sum_{j\in {\mathcal S}}\beta_j\widetilde{\bm{T}}_{g,j}=(\mathbf{ T}^k)_{{\mathcal V}{\mathcal S}}\bm{\beta}=(\mathbf{ T}^k)_{{\mathcal V}{\mathcal S}}((\mathbf{ T}^k)_{{\mathcal S}})^{-1}\bm{f}_{\mathcal S},
\label{eq3}
\end{equation}
where $\bm{\beta}=((\mathbf{ T}^k)_{{\mathcal S}})^{-1}\bm{f}_{\mathcal S}$.
\begin{theorem} The $|\mathcal{F}|$-bandlimited signals with $|\mathcal{F}|\leq|\mathcal{S}|$ are perfectly recovered with \eqref{eq3} if $k$ becomes large, i.e.,
\begin{equation}
\lim_{k\rightarrow \infty}\widehat{\bm{f}}_k=\bm{f}
\end{equation}
 as long as the kernel of the localization operator satisfies $g(\lambda_i)>g(\lambda_j)$ for all $\lambda_i<|\mathcal{F}|$ and $\lambda_j\ge|\mathcal{F}|$.
\end{theorem} 
\begin{proof} 
$\widehat{\bm{f}}_k$ can be rewritten as
\begin{equation}
\begin{split}
\widehat{\bm{f}}_k=&(\mathbf{ T}^k)_{{\mathcal V}{\mathcal S}}((\mathbf{ T}^k)_{{\mathcal S}})^{-1}\bm{f}_{\mathcal S}\\
=&(\mathbf{ U}g^k(\mathbf{ \Lambda})\mathbf{ U}^*)_{{\mathcal V}{\mathcal S}}((\mathbf{ U}g^k(\mathbf{ \Lambda})\mathbf{ U}^*)_{{\mathcal S}})^{-1}\bm{f}_{\mathcal S}\\
=&\mathbf{ U}g^k(\mathbf{ \Lambda})\mathbf{ U}_{{\mathcal S}{\mathcal V}}^*(\mathbf{ U}_{{\mathcal S}{\mathcal V}}g^k(\mathbf{ \Lambda})\mathbf{ U}_{{\mathcal S}{\mathcal V}}^*)^{-1}\bm{f}_{\mathcal S}\\
=&\mathbf{ U}g^{k/2}(\mathbf{ \Lambda})(\mathbf{ U}_{{\mathcal S}{\mathcal V}}g^{k/2}(\mathbf{ \Lambda}))^+\bm{f}_{\mathcal S}\\
:=&\mathbf{ U}g^{k/2}(\mathbf{ \Lambda})\widetilde{\bm{\alpha}}_k,
\label{eq:mod}
\end{split}
\end{equation}
where $\widetilde{\bm{\alpha}}_k:=(\mathbf{ U}_{{\mathcal S}{\mathcal V}}g^{k/2}(\mathbf{ \Lambda}))^+\bm{f}_{\mathcal S}$ and it is the estimation of $\bm{\alpha}_k=(g^{k/2}(\mathbf{ \Lambda}))^+\mathbf{ U}^*\bm{f}$ that is the modified graph Fourier coefficients by $g^{k/2}({\lambda})$. Because $\bm{f}$ is $|\mathcal{F}|$-bandlimited, $\alpha_k(m)=0$ for $m>|\mathcal{F}|$ is always satisfied.　
Because $\bm{f}_{{\mathcal S}}=\mathbf{ U}_{{\mathcal S}{\mathcal V}}g^{k/2}(\mathbf{ \Lambda})\bm{\alpha}_k$, $\widetilde{\bm{\alpha}}_k=(\mathbf{ U}_{\mathcal{S}\mathcal{V}}g^{k/2}(\mathbf{ \Lambda}))^+\bm{f}_\mathcal{S}$ is the estimation of $\bm{\alpha}_k$ only from $\bm{f}_\mathcal{S}$. 

The calculation of the pseudoinverse usually causes an estimation error. However, if the kernel satisfies $g(\lambda_i)>g(\lambda_j)$ for all $\lambda_i\leq|\mathcal{F}|$ and $\lambda_j>|\mathcal{F}|$, the error can be ignored. Here, \eqref{eq:mod} is equivalently rewritten as 
\begin{equation}
\widehat{\bm{f}}_k=\mathbf{ U}\left(\frac{g(\mathbf{ \Lambda})}{\beta}\right)^{k/2}\left(\mathbf{ U}_{\mathcal{S}\mathcal{V}}\left(\frac{g(\mathbf{ \Lambda})}{\beta}\right)^{k/2}\right)^+\bm{f}_{\mathcal S},
\label{eqn:f_hat}
\end{equation}
where ${\beta}=\min_{0\leq i\leq|\mathcal{F}|-1}\ g(\lambda_i)$. 
Because $g(\lambda_i)/\beta<1$ is satisfied for all $i\ge|\mathcal{F}|$, 
$\lim_{k\rightarrow \infty}(g(\lambda_i)/\beta)^{k/2}\rightarrow 0$.
Then, for a sufficiently large $k$,
\begin{equation}
\mathbf{ U}_{\mathcal{S}\mathcal{V}}\left(\frac{g(\mathbf{ \Lambda}_\mathcal{F})}{\beta}\right)^{k/2}\approx
\begin{bmatrix}
\left(\frac{g(\mathbf{ \Lambda})}{\beta}\right)^{k/2}\mathbf{ U}_{\mathcal{S}\mathcal{F}}^T &\mathbf{ 0}_{|\mathcal{S}|}
\end{bmatrix}^T
\end{equation}
and 
\begin{equation}
\left(\mathbf{ U}_{\mathcal{S}\mathcal{V}}\left(\frac{g(\mathbf{ \Lambda})}{\beta}\right)^{k/2}\right)^+\approx
\begin{bmatrix}
\left(\frac{g(\mathbf{ \Lambda}_\mathcal{F})}{\beta}\right)^{-k/2}\mathbf{ U}_{\mathcal{S}\mathcal{F}}^{+} &\mathbf{ 0}_{|\mathcal{S}|}
\end{bmatrix},
\label{eqn:pinv}
\end{equation} 
where $\mathbf{ 0}_{|\mathcal{S}|}$ is a $|\mathcal{S}|\times|\mathcal{S}|$ null matrix. 
From \eqref{eqn:pinv}, \eqref{eqn:f_hat} can be rewritten as
\begin{equation}
\widehat{\bm{f}}_k=\mathbf{ U}_{\mathcal{V}\mathcal{F}}\mathbf{ U}_{\mathcal{S}\mathcal{F}}^{+} 
\bm{f}_{\mathcal S}.
\end{equation}
This coincides with \eqref{rec}, and therefore, \eqref{eq3} can perfectly recover the $|\mathcal{F}|$-bandlimited signals.
\end{proof}

\subsection{Reconstruction Error}
Here, we consider the reconstruction error by \eqref{eq3} when the sampled signal contains additive noise, i.e., $\bm{o}=\bm{f}_{\mathcal{S}}+\bm{n}_{\mathcal{S}}$, where $\bm{f}$ is an $|\mathcal{F}|$-bandlimited signal and the additive noise $\bm{n}\in{\mathbb R}^{N}$ is i.i.d. and zero-mean.

The error $\bm{e} := \bm{f} - \widehat{\bm{f}}_k$ is represented as
\begin{equation}
\bm{e}=\bm{f}-(\mathbf{ T}^k)_{{\mathcal V}{\mathcal S}}((\mathbf{ T}^k)_{{\mathcal S}})^{-1}\bm{o}=(\mathbf{ T}^k)_\mathcal{VS}((\mathbf{ T}^k)_\mathcal{S})^{-1}\bm{n}_\mathcal{S}.
\end{equation}
This can be rewritten as
\begin{equation}
\begin{split}
\bm{e}=&\mathbf{ U}g^k(\mathbf{ \Lambda})\mathbf{ U}_{{\mathcal S}{\mathcal V}}^*(\mathbf{ U}_{{\mathcal S}{\mathcal V}}g^k(\mathbf{ \Lambda})\mathbf{ U}_{{\mathcal S}{\mathcal V}}^*)^{-1}\bm{n}_{\mathcal S}\\
=&\mathbf{ U}g^{k/2}(\mathbf{ \Lambda})\mathbf{ U}^*\mathbf{ U}g^{k/2}(\mathbf{ \Lambda})\mathbf{ U}_{{\mathcal S}{\mathcal V}}^*\\
&\times(\mathbf{ U}_{{\mathcal S}{\mathcal V}}g^{k/2}(\mathbf{ \Lambda})\mathbf{ U}^*\mathbf{ U}g^{k/2}(\mathbf{ \Lambda})\mathbf{ U}_{{\mathcal S}{\mathcal V}}^*)^{-1}\bm{n}_{\mathcal S}\\
=&\mathbf{ T}^{k/2}((\mathbf{ T}^{k/2})_{{\mathcal S}{\mathcal V}})^+\bm{n}_{\mathcal S}.
\end{split}
\end{equation}

Then, the error covariance matrix is calculated as
\begin{equation}
\begin{split}
\mathbf{ E}=&\bm{e}\bm{e}^*\\
=&\mathbf{ T}^{k/2}((\mathbf{ T}^{k/2})_{{\mathcal S}{\mathcal V}})^+\bm{n}_{\mathcal S}\bm{n}_{\mathcal S}^*((\mathbf{ T}^{k/2})_{{\mathcal S}{\mathcal V}}^*)^+\mathbf{ T}^{k/2}\\
=&\mathbf{ T}^{k/2}((\mathbf{ T}^{k/2})_{{\mathcal S}{\mathcal V}}^*(\mathbf{ T}^{k/2})_{{\mathcal S}{\mathcal V}})^+\mathbf{ T}^{k/2}.\\
\end{split}
\label{eqn:errorcov}
\end{equation}
For minimizing the error, we should minimize or maximize the trace, determinant, or maximum eigenvalue of $\mathbf{E}$, depending on the optimization strategy. The approaches are summarized in Table \ref{tab:od}, and their derivations are shown in Appendix C.

\begin{table}[t]
  \centering 
  \caption{Proposed Vertex Selection Based on Minimization of Error Covariance Matrix.}
  \begin{tabular}{c|l}
\hline
Optimal Design & \multicolumn{1}{c}{Objective}\\\hline
A-optimal& $\argmin_{\mathcal{S}\subset\mathcal{V}:|\mathcal{S}|=F} \text{tr}[((\mathbf{ T}^{k})_{{\mathcal S}})^{-1}]$\\
D-optimal & $\argmin_{\mathcal{S}\subset\mathcal{V}:|\mathcal{S}|=F} \det[((\mathbf{ T}^{k})_{{\mathcal S}})^{-1}]$\\
E-optimal & $\argmin_{\mathcal{S}\subset\mathcal{V}:|\mathcal{S}|=F} \|((\mathbf{ T}^{k/2})_{{\mathcal S}{\mathcal V}})^+\|_2$\\
T-optimal & $\argmax_{\mathcal{S}\subset\mathcal{V}:|\mathcal{S}|=F} \text{tr}[(\mathbf{ T}^{k})_{{\mathcal S}}]$\\
\hline
  \end{tabular}
  \label{tab:od}
  \end{table}

\subsection{Vertex Selection Methods Based on Covering Area of Localization Operator}
While the cost function based on \eqref{eqn:errorcov} can interpret various existing SSS approaches as its special cases  (see Section \ref{sec:rela}), a naive realization of maximizing/minimizing the cost functions in Table \ref{tab:od} needs eigendecomposition, which leads to high computational complexity.

We reconsider the intuition of the localization operator, which avoids the abovementioned problem. Intuitively, the set of vertices $\mathcal{S}$ should be the most informative concerning the localization operator. Each localization operator $\bm{T}_{g,i}$ would be regarded as the area where the $i$th vertex can estimate unobserved signal values. Therefore, we select vertices such that $\bm{T}_{g,i}\ (i\in \mathcal{S})$ covers the entire area evenly, i.e., the sum of $\|\bm{T}_{g,i}\|_2^2\ (i\in \mathcal{S})$ is large and the overlapping area covered by both $\bm{T}_{g,i}$ and $\bm{T}_{g,j}\  (i\neq j)$ is small.

Such a set is obtained by optimizing the following function:
\begin{equation}
\begin{split}
\mathcal{S}^*&=\argmax_{\mathcal{S}\subset\mathcal{V}:|\mathcal{S}|=F} \sum_{i\in \mathcal{S}} \langle\bm{T}_{g,i},\bm{T}_{g,i}\rangle-\sum_{j\in \mathcal{S}, j\neq i} \left\langle|\bm{T}_{g,i}|, |{\bm{T}_{g,j}}|\right\rangle\\
&=\argmax_{\mathcal{S}\subset\mathcal{V}:|\mathcal{S}|=F} \sum_{i\in \mathcal{S}} \left\langle\left(|\bm{T}_{g,i}|-\sum_{ j\in \mathcal{S}, j\neq i}|{\bm{T}_{g,j}}|\right),|\bm{T}_{g,i}|\right\rangle,
\label{eq:obj}
\end{split}
\end{equation}
To optimize the cost function, we use a greedy algorithm, which appends one vertex in the $m$th iteration by selecting a vertex $y^*$ satisfying the following function:
\begin{equation}
y^*=
\argmax_{y\in \mathcal{S}_m^c}\ \left\langle  R\left(\eta\mathbf{ 1}_{N\times 1}-\sum_{j\in\mathcal{S}_m}|{\bm{T}_{g,j}}|\right),|\bm{T}_{g,y}|\right\rangle,
\label{eq:cost2}
\end{equation}
where $R(\cdot)$ is the ramp function that satisfies $[R(\bm{x})](i)=x(i)$ if $x(i)\ge0$ and $0$ otherwise, and  $\eta\in\mathbb{R}_+$ is an arbitrary real value.

In \eqref{eq:cost2}, we calculate the weighted norm of $\bm{T}_{g,y}$. A small weight is assigned to ${T}_{g,y}(i)$ if the $i$th vertex has already been covered: In this case, the weight of $\sum_{j\in \mathcal{S}_m}|\bm{T}_{g,j}|$ at the $i$th vertex is large. In each iteration, we avoid selecting vertices whose localization operators overlap with those of already-selected vertices, because the weight for ${T}_{g,y}(i)$ becomes $0$ when $\sum_{j\in\mathcal{S}_m}|{T}_{g,j}(i)|\ge \eta$.
In this study, we use
$\eta=\frac{1}{|\mathcal{V}|}\sum_{i\in\mathcal{V}}\sum_{j\in\mathcal{S}_m}|T_{g,j}(i)|,$
which is experimentally determined.

If the kernel $g(\lambda)$ is a polynomial function, we can calculate \eqref{eq:cost2} without an eigendecomposition of the graph Laplacian. This is because \eqref{loc} is rewritten as $\mathbf{ T}=\sqrt{N}g(\mathbf{ L})$ when $g(\lambda)$ is a polynomial function. Therefore, localization operators can be obtained without the eigenvectors themselves. As a result, if the original kernel $g(\lambda)$ is a polynomial or the Chebyshev polynomial approximation is applied to $g(\lambda)$, an eigendecomposition is not required for the proposed SSS. Using the polynomial function, \eqref{eq:cost2} can be rewritten as
\begin{equation}
y^*=\argmax_{y\in \mathcal{S}_m^c}\ \left[\ \!\left|\mathbf{ T}\right|\bm{w}\ \!\!\right](y),
\label{eq:cost3}
\end{equation}
where $\bm{w}=R\left(\eta\mathbf{ 1}_{N\times 1}-\sum_{j\in\mathcal{S}_m}|{\bm{T}_{g,j}}|\right)$. In particular, when $g(\cdot)$ is a heat kernel, i.e., $g(\lambda)=\exp(-s\lambda)$ for some constant $s>0$, all elements in $\bm{T}_{g,j}$ have nonnegative values \cite{Hammon2011}. Therefore, we need not calculate the absolute value of each element in the localization operators.

\subsection{Computational Complexity}

  
  \begin{table*}[t]
  \centering 
  \caption{Computational Complexities of Graph Signal Processing-based Approaches}\label{t1}
  \begin{tabular}{c|c|c|c|c|c|c}
\hline
& MaxCutoff \cite{Anis2016}	&	MinSpec	\cite{Chen2015}	&	MinFrob\cite{Tsitsv2016}	&	MaxFrob\cite{Tsitsv2016}	&	MaxPVol\cite{Tsitsv2016}&Proposed Method w/ CPA\\\hline
Eigen-pair or operator computations &$O(k|\mathcal{E}|FT_1)$&\multicolumn{4}{c|}{$O((|\mathcal{E}|F+CF^3)T_F)$}&$O(|\mathcal{E}|NP)$\\\hline
Sampling set search&$O(NF)$&$O(NF^4)$&$O(NF^4)$&$O(NF)$&$O(F^3)$&$O(JF)$\\\hline\hline
Section&III-B-1&III-B-2&III-B-3-i&III-B-3-ii&III-B-3-iii&IV-B-2\\\hline
  \end{tabular}
  \label{table:cc}
  \end{table*}

Table \ref{table:cc} compares the computational complexities of the graph signal processing-based methods\cite{Anis2014,Anis2016,Chen2015} and the proposed method shown in Section IV-B-2, where $T_1$ is the average number of iterations required for the convergence of a single eigen-pair, $T_F$ is the number of iterations of convergence for the first $F$ eigen-pair, $k$ provides a trade-off between performance and complexity of the method proposed in \cite{Anis2016}, $C$ is a constant, $P$ is the approximation order of the Chebyshev polynomial approximation, and $J$ is the number of nonzero elements in $\mathbf{T}$. We follow the notation in \cite{Anis2016}. 

Note that $\mathbf{T}$ is a sparse matrix because its $i$th row has nonzero elements only at the columns corresponding to the $i$th vertex and several neighboring vertices.
The calculation of the localization operator in the proposed method includes complexity for performing the Chebyshev polynomial approximation and filtering\cite{Hammon2011}.

It can be seen that the calculation of the localization operator shows much lower complexity than  those for calculating the eigen-pairs in the other approaches. MaxCutoff needs the calculation of eigen-pairs in each iteration whereas the other methods calculate the eigen-pairs or operator only once. Therefore, although the proposed method has a higher complexity order than MaxCutoff in the sampling set search in Table \ref{table:cc}, its total execution time is usually lower than MaxCutoff.

\section{Relationship Between Proposed and Conventional Methods}
\label{sec:rela}
The objective functions of the conventional methods can be rewritten using the localization operator with various kernels. They are summarized in Table \ref{tab:con}. Furthermore, we show that the existing approaches based on the graph Fourier basis are one of the proposed method based on the error minimization.

\subsection{Sensor Position Selection Based on Gaussian Process}
The sensor selection methods based on the GP model introduced in Section III-A can be viewed as the SSS approaches for graph signals that use the covariance matrix instead of the Laplacian matrix. In general, the graph Laplacian (precision matrix) and the covariance matrix have the following relationship\cite{Lawren2012}:
\begin{equation}
\mathbf{ L}=\mathbf{ K}^{-1}-\delta\mathbf{ I}.
\label{eq:invc}
\end{equation} 
The parameter $\delta$ prevents the precision matrix from being singular. The precision matrix has the same set of eigenvectors $\{\bm{u}_i=\bm{v}_i(\mathbf{ K})\}_{i=0,\ \ldots,\ N-1}$ with corresponding eigenvalues $\{\lambda_i=\frac{1}{\mu_i(\mathbf{ K})}-\delta\}_{i=0,\ \ldots,\ N-1}$. 

From \eqref{eq:invc}, \eqref{eq:dis1} indicates that the random signals have following distributions:
\begin{equation}
\begin{split}
p(\bm{f})&\propto\exp\left(-\bm{f}^T\mathbf{ K}^{-1}\bm{f}\right)\\
&=\exp\left(-\bm{f}^T(\mathbf{ L}+\delta\mathbf{ I})\bm{f}\right)\\
&= \exp\left(-\sum_{i}\sum_{j}A({i,j})({f}(i)-{f}(j))^2-\delta\sum_{i}{f}(i)^2\right),\\
\label{eq:dis2}
\end{split}
\end{equation}
namely, vertices with similar signal values are connected by edges with large weights. This also indicates that the signals are smooth over the graph with the Laplacian obtained by \eqref{eq:invc}. 

The localization operator can rewrite the entropy criterion \eqref{eq:en} as
\begin{equation}
\begin{split}
\mathcal{S}^*&=\argmax_{\mathcal{S}\subset\mathcal{V}:|\mathcal{S}|=F} \log\det[ \mathbf{ K}_\mathcal{S}]\\
&=\argmax_{\mathcal{S}\subset\mathcal{V}:|\mathcal{S}|=F}\det[({(\mathbf{ L}+\delta\mathbf{ I})^{-1}})_\mathcal{S}]\\
&=\argmax_{\mathcal{S}\subset\mathcal{V}:|\mathcal{S}|=F}\det[ \mathbf{ T}^\text{K}_\mathcal{S}],\\
\end{split}
\end{equation}
where $\mathbf{ T}^\text{K}=(\mathbf{ L}+\delta\mathbf{ I})^{-1}$, i.e., the localization operator with the kernel $g(\lambda)=1/(\lambda+\delta)$.
Similarly, \eqref{eq:mi1} is
\begin{equation}
\begin{split}
\mathcal{S}^*&=\argmax_{\mathcal{S}\subset\mathcal{V}:|\mathcal{S}|=F} \log\det[\mathbf{ K}_\mathcal{S}]+\log\det[\mathbf{ K}_\mathcal{\widebar{S}\widebar{S}}]\\
&=\argmax_{\mathcal{S}\subset\mathcal{V}:|\mathcal{S}|=F}\ \det[\mathbf{ T}^\text{K}_\mathcal{S}]\det[\mathbf{ T}^\text{K}_\mathcal{\widebar{S}\widebar{S}}].
\end{split}
\end{equation}

Next, we clarify the characteristic of the greedy optimization step from a graph signal processing perspective.
From the block matrix inversion formula, the inversion of the covariance matrix can be represented as\cite{Gadde2015}:
\begin{equation}
\begin{split}
\mathbf{ K}^{-1}&=\begin{bmatrix}
\mathbf{ K}_{\mathcal{S}^c} & \mathbf{ K}_{\mathcal{S}^c\mathcal{S}}\\
\mathbf{ K}_{\mathcal{S}\mathcal{S}^c} & \mathbf{ K}_{\mathcal{S}}
\end{bmatrix}^{-1}\\
&=\begin{bmatrix}
\mathbf{ K}_{\mathcal{S}^c|\mathcal{S}}^{-1} & -(\mathbf{ K}_{{\mathcal{S}}^c})^{-1}\mathbf{ K}_{\mathcal{S}^c\mathcal{S}}\mathbf{ K}^{-1}_{\mathcal{S}|\mathcal{S}^c}  \\
-(\mathbf{ K}_{\mathcal{S}})^{-1}\mathbf{ K}_{\mathcal{S}^c\mathcal{S}}^T\mathbf{ K}^{-1}_{\mathcal{S}^c|\mathcal{S}}& \mathbf{ K}_{\mathcal{S}|\mathcal{S}^c}^{-1}\\
\end{bmatrix},
\end{split}
\label{eq:invk}
\end{equation}
where $\mathbf{ K}_{\mathcal{S}^c|\mathcal{S}}=\mathbf{ K}_{\mathcal{S}^c}-\mathbf{ K}_{\mathcal{S}^c\mathcal{S}}(\mathbf{ K}_{\mathcal{S}})^{-1}\mathbf{ K}_{\mathcal{S}^c\mathcal{S}}^T$ and $\mathbf{ K}_{\mathcal{S}|\mathcal{S}^c}=\mathbf{ K}_{\mathcal{S}}-\mathbf{ K}_{\mathcal{S}\mathcal{S}^c}(\mathbf{ K}_{\mathcal{S}^c})^{-1}\mathbf{ K}_{\mathcal{S}\mathcal{S}^c}^T$.
Using \eqref{eq:invc} and \eqref{eq:invk}, the graph Laplacian and the covariance matrix have the following relationship:
\begin{equation}
\mathbf{ L}_{\mathcal{S}^c}+\delta\mathbf{I}=(\mathbf{ K}_{\mathcal{S}^c}-\mathbf{ K}_{\mathcal{S}^c\mathcal{S}}(\mathbf{ K}_{\mathcal{S}})^{-1}\mathbf{ K}_{\mathcal{S}^c\mathcal{S}}^T)^{-1}.
\label{eq:12}
\end{equation}

Fig. \ref{fig:L} considers a toy example that uses a synthesized simple graph for the sake of clarity. From \eqref{eq:12}, we can rewrite the entropy criterion in \eqref{eq:en3} as:
\begin{equation}
\begin{split}
y^*\leftarrow\argmax_{y\in \mathcal{S}_m^c}\frac{1}{{L}^y(y,y)+\delta^y},
\label{eq:eng}
\end{split}
\end{equation}
where $\mathbf{L}^y$ is the Laplacian matrix of the graph with the vertices ${\mathcal{S}_m}\cup y$ and the edges between these vertices (Fig. \ref{fig:L} (b)), and $\delta^y$ is the variance of $\bm{f}_{\mathcal{S}_m\cup y}$. It can be seen that the entropy criterion selects a vertex that has the minimum degree with the selected vertices, i.e., the vertex with the weakest connection with the already-selected vertices is selected. Because of this, the entropy criterion often places many vertices at the corners or boundaries of the space, as is well known. 

The MI criterion in \eqref{eq:mi3} can also be rewritten as
\begin{equation}
y^*\leftarrow
\argmax_{y\in \mathcal{S}_m^c}\frac{\widebar{L}^y(y,y)+\widebar\delta^y}{{L}^y(y,y)+\delta^y},
\end{equation}
where $\widebar{\mathbf{L}}^y$ is the graph Laplacian containing the unselected vertices ${\mathcal S}_m^c$ and the edges in ${\mathcal S}_m^c$ (Fig. \ref{fig:L} (c)), and $\widebar\delta^y$ is the variance of $\bm{f}_{\mathcal{S}_m^c}$. It can be observed that the MI criterion chooses the vertex that has the weakest connection with the selected vertices and strongest connection with the unselected vertices.

\begin{figure}[tp]
  \centering
  \subfigure[][]{\includegraphics[width=.25\linewidth]{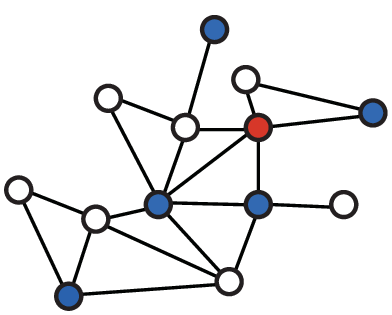}}\ \ 
    \subfigure[][]{\includegraphics[width=.25\linewidth]{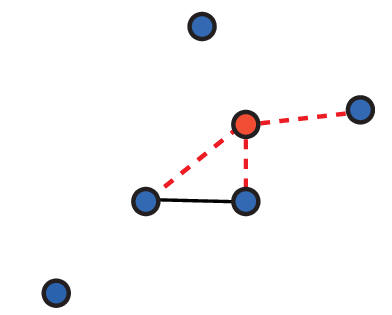}}\ \ 
  \subfigure[][]{\includegraphics[width=.25\linewidth]{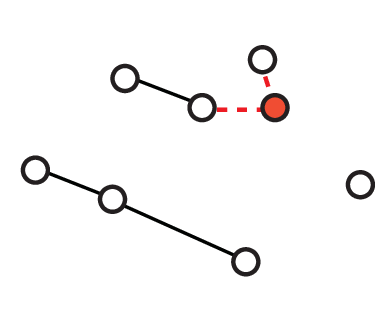}}
  \caption{(a) Original graph. The blue vertices and red vertex indicate $\mathcal{S}$ and $y$, respectively. (b) $\mathbf{L}^y$. $L^y(y,y)$ is the total weight of the red dashed edges. (c) $\bar{\mathbf{L}}^y$. $\bar{L}^y(y,y)$ is the total weight of the red dashed edges. }
  \label{fig:L}
\end{figure}

It is worth noting that the conventional entropy and MI criterion select vertices according to the edge information in the \textit{graph vertex domain}, whereas the sampling methods based on the graph Fourier basis including the proposed method select vertices while considering the spectrum in the \textit{graph frequency domain}.

\subsection{Graph Sampling Based on Fourier Basis}
\subsubsection{Based on Cutoff Frequency (MaxCutoff)}
The reconstruction algorithm in \eqref{eq3} with the ideal kernel, i.e., $\mathbf{ T}=\mathbf{ U}\text{diag}(\bm{1}_{\mathcal{F}})\mathbf{ U}^*$, coincides with that of the conventional methods in \eqref{rec}:
\begin{equation}
\begin{split}
\widehat{\bm{f}}_k=&(\mathbf{ T}^k)_{{\mathcal V}{\mathcal S}}((\mathbf{ T}^k)_{{\mathcal S}})^{-1}\bm{f}_{\mathcal S}\\
=&\mathbf{ U}\text{diag}(\bm{1}_{\mathcal{F}})\mathbf{ U}_{{\mathcal S}{\mathcal V}}^*(\mathbf{ U}_{{\mathcal S}{\mathcal V}}\text{diag}(\bm{1}_{\mathcal{F}})\mathbf{ U}_{{\mathcal S}{\mathcal V}}^*)^{-1}\bm{f}_{\mathcal S}\\
=&\mathbf{ U}_{{\mathcal V}{\mathcal F}}\mathbf{ U}_{{\mathcal S}{\mathcal F}}^*(\mathbf{ U}_{{\mathcal S}{\mathcal F}}\mathbf{ U}_{{\mathcal S}{\mathcal F}}^*)^{-1}\bm{f}_{\mathcal S}\\
=&\mathbf{ U}_{{\mathcal V}{\mathcal F}}\mathbf{ U}_{{\mathcal S}{\mathcal F}}^+\bm{f}_{\mathcal S}.
\end{split}
\end{equation}
Furthermore, the objective function in \eqref{gs1} can be rewritten as
\begin{equation}
\begin{split}
\mathcal{S}^*=&\argmax_{\mathcal{S}\subset\mathcal{V}:|\mathcal{S}|=F}\mu_\text{min}(\mathbf{ L}^k)_{\mathcal{S}^c}\\
=&\argmin_{\mathcal{S}\subset\mathcal{V}:|\mathcal{S}|=F}\mu_\text{max}(((\mathbf{ L}^k)_{\mathcal{S}^c})^{-1})\\
=&\argmin_{\mathcal{S}\subset\mathcal{V}:|\mathcal{S}|=F}\|((({\mathbf{ T}^{\text{L}}})^{k})_{\mathcal{S}^c})^{-1}\|_2,\\
\end{split}
\label{conv1}
\end{equation}
where $\mathbf{ T}^{\text{L}}=\mathbf{ L}$ is the localization operator matrix with $g(\lambda_i)=\lambda_i$. 

\subsubsection{Based on Error Minimization}
\begin{itemize}
\item\textit{MinSpec:}
\eqref{gs2-1} is rewritten as
\begin{equation}
\begin{split}
\mathcal{S}^*&=\argmax_{\mathcal{S}\subset\mathcal{V}:|\mathcal{S}|=F}\sigma_\text{min}(\mathbf{ U}_{\mathcal{S}\mathcal{F}})\\
&=\argmax_{\mathcal{S}\subset\mathcal{V}:|\mathcal{S}|=F}\sigma_\text{min}(\mathbf{ D}_\text{ver}\mathbf{ U}\text{diag}(\bm{1}_{\mathcal{F}}))\\
&=\argmax_{\mathcal{S}\subset\mathcal{V}:|\mathcal{S}|=F}\sigma_\text{min}(\text{diag}(\bm{1}_{\mathcal{F}})\mathbf{ U}^*\mathbf{ D}_\text{ver})\\
&=\argmax_{\mathcal{S}\subset\mathcal{V}:|\mathcal{S}|=F}\sigma_\text{min}(\mathbf{ U}\text{diag}(\bm{1}_{\mathcal{F}})\mathbf{ U}^*\mathbf{ D}_\text{ver})\\
&=\argmax_{\mathcal{S}\subset\mathcal{V}:|\mathcal{S}|=F}\sigma_\text{min}(\mathbf{ T}^\text{I}_\mathcal{VS})=\|\mathbf{ T}^\text{I}_\mathcal{VS}\|_2,
\end{split}
\label{conv2}
\end{equation}
where $\mathbf{ T}^\text{I}=\mathbf{ U}\text{diag}(\bm{1}_{\mathcal{F}})\mathbf{ U}^*$ which is the localization operator matrix with the ideal filter: $g(\lambda_i)=1$ for $\lambda_i\in\mathcal{F}$ and $0$ otherwise.
\item\textit{MinTrac:}\eqref{gs2-2} is also rewritten as
\begin{equation}
\begin{split}
\mathcal{S}^*&=\argmin_{\mathcal{S}\subset\mathcal{V}:|\mathcal{S}|=F}\text{tr}[(\mathbf{ U}_{\mathcal{S}\mathcal{F}}\mathbf{ U}_{\mathcal{S}\mathcal{F}}^*)^{-1}]\\
&=\argmin_{\mathcal{S}\subset\mathcal{V}:|\mathcal{S}|=F}\text{tr}[(\mathbf{ T}^\text{I}_\mathcal{S})^{-1}].\\
\end{split}
\label{conv2}
\end{equation}
\end{itemize}

\subsubsection{Based on Localized Basis}
\begin{enumerate}
\renewcommand{\labelenumi}{(\roman{enumi})}
\item \textit{MinFrob:} \eqref{gs3-1} is rewritten as 
\begin{equation}
\begin{split}
\mathcal{S}^*&=\argmin_{\mathcal{S}\subset\mathcal{V}:|\mathcal{S}|=F}\|(\text{diag}(\bm{1}_{\mathcal{F}})\mathbf{ U}^*\mathbf{ D}_\text{ver})^+\|_F\\
&=\argmin_{\mathcal{S}\subset\mathcal{V}:|\mathcal{S}|=F}\sum_{i=0}^{|\mathcal{F}|}\frac{1}{\sigma_i(\text{diag}(\bm{1}_{\mathcal{F}})\mathbf{ U}^*\mathbf{ D}_\text{ver})}\\
&=\argmin_{\mathcal{S}\subset\mathcal{V}:|\mathcal{S}|=F}\sum_{i=0}^{|\mathcal{F}|}\frac{1}{\sigma_i(\mathbf{ T}^\text{I}_\mathcal{VS})}\\
&=\argmin_{\mathcal{S}\subset\mathcal{V}:|\mathcal{S}|=F}\|(\mathbf{ T}^\text{I}_\mathcal{VS})^+\|_F\\
&=\argmin_{\mathcal{S}\subset\mathcal{V}:|\mathcal{S}|=F}\text{tr}[(\mathbf{ T}^\text{I}_\mathcal{S})^{-1}].
\end{split}
\label{conv3}
\end{equation}
\item \textit{MaxFrob:} Similar to (i), \eqref{gs3-2} is rewritten as
\begin{equation}
\begin{split}
\mathcal{S}^*&=\argmax_{\mathcal{S}\subset\mathcal{V}:|\mathcal{S}|=F} \|\text{diag}(\bm{1}_{\mathcal{F}})\mathbf{ U}^*\mathbf{ D}_\text{ver}\|_F\\
&=\argmax_{\mathcal{S}\subset\mathcal{V}:|\mathcal{S}|=F}\|\mathbf{ T}_\mathcal{VS}^\text{I}\|_F\\
&=\argmax_{\mathcal{S}\subset\mathcal{V}:|\mathcal{S}|=F}\text{tr}[\mathbf{ T}^\text{I}_\mathcal{S}].\\
\end{split}
\label{conv4}
\end{equation}
\item \textit{MaxPVol:} \eqref{gs3-3} is rewritten as
\begin{equation}
\begin{split}
\mathcal{S}^*&=\argmax_{\mathcal{S}\subset\mathcal{V}:|\mathcal{S}|=F} \det[\mathbf{ U}_{\mathcal{SF}}\mathbf{ U}_{\mathcal{SF}}^*]\\
&=\argmax_{\mathcal{S}\subset\mathcal{V}:|\mathcal{S}|=F}\det[\mathbf{ D}_\text{ver}\mathbf{ U}\text{diag}(\bm{1}_{\mathcal{F}})(\mathbf{ D}_\text{ver}\mathbf{ U}\text{diag}(\bm{1}_{\mathcal{F}}))^*]\\
&=\argmax_{\mathcal{S}\subset\mathcal{V}:|\mathcal{S}|=F}\det[\mathbf{ T}_\mathcal{S}^\text{I}].\\
\end{split}
\label{conv5}
\end{equation}
\end{enumerate}

The objective functions represented by the localization operators are summarized in Table \ref{tab:con}. Because $(\mathbf{T}^\text{I})^k=\mathbf{T}^\text{I}$, the conventional SSS methods shown in \eqref{conv2}--\eqref{conv5} coincide with the proposed SSS based on the error minimization, which is introduced in Table \ref{tab:od}, in the case of using the ideal kernel for the localization operator. Furthermore, MaxCutoff, as shown in \eqref{conv1}, can also be viewed as the objective function for minimizing the error covariance matrix caused by the reconstruction shown in \eqref{eq3}. 
\begin{proposition}
MaxCutoff can be viewed as the error minimization for signal reconstruction using \eqref{eq3} with $\mathbf{T}=(\mathbf{L}+\delta{\bf I})^{-1}$, in the case in which $\delta$ goes to zero.\footnote{$\delta$ prevents the precision matrix from being singular.}.
\end{proposition}
\noindent The proof is shown in Appendix D.

In summary, all of the existing SSS methods introduced in this paper can be viewed as special cases of the proposed SSS based on the error minimization with the different optimal criteria and kernels. 

\section{Experimental Results}
\label{sec:exp}

\subsection{Setup}
In the experiments, we used the following six graphs:
\begin{itemize}
    \item Random sensor graph.
    \item Random graph with Erd\H{o}s--R\'{e}nyi model (ER graph): The edge connecting probability was set to $0.05$.
    \item Random regular graph: Each vertex connects to six vertices.
    \item Random graph with Barab\'{a}si--Albert model (BA graph): The initial connected graph has six vertices.
    \item Community graph: $11$ communities with random sizes are yielded.
    \item Minnesota Traffic graph.
\end{itemize}
For the comparison of execution time, we selected the random sensor graph with different numbers of vertices. For the comparison of prediction errors, we used all six graphs with the number of vertices for random graphs is set to $N = 500$ and that for the Minnesota Traffic graph is $N = 2642$.

The performance of the proposed method is compared with the following approaches:
\begin{itemize}
    \item GP-based methods: Entropy- and MI-based criteria (abbreviated as Entropy and MI) \cite{Cressi2015, Shewry1987, Krause2008}
    \item Graph-based SSS methods with deterministic selection \cite{Chen2015, Anis2016, Tsitsv2016}
    \item Graph-based SSS using random sampling with nonuniform sampling probability distribution (abbreviated as RandSamp) \cite{Puy2018}
\end{itemize}
For the GP-based methods, we need to estimate the covariance matrix. Based on \eqref{eq:invc}, we simply set $\mathbf{K} = (\mathbf{L} + \delta\mathbf{I})^{-1}$, where $\delta = 0.01$. MaxCutoff needs the parameter $k$ for estimating the cutoff frequency, which is set to $k=14$. RandSamp also needs the parameter $\gamma$ for the reconstruction, which is set to $\gamma = 1$. Although the optimal $\gamma$ widely varies according to the graph used, $\gamma = 1$ is used for one of the experiments in \cite{Puy2018}. 

As the signal model for a realistic situation, we use noisy bandlimited signals where
\begin{equation}
    \widebar{\bm{f}} = [\widebar{\bm{f}}_{\text{bl}}^T,\ \mathbf{0}_{N-|\mathcal{F}|}^T]^T + \bm{\epsilon},
\end{equation}
in which $\widebar{\bm{f}}_{\text{bl}} \in \mathbb{R}^{|\mathcal{F}|}$ is a random vector of length $|\mathcal{F}|$ whose elements conform with $\mathcal{N}(0, 0.2)$, and $\bm{\epsilon} \in \mathbb{R}^N$ is an iid noise vector following $\mathcal{N}(0, 5 \times 10^{-3})$. We set $|\mathcal{F}| = 100$ for all the experiments.

\begin{figure}[t]
\centering
\includegraphics[width=\linewidth]{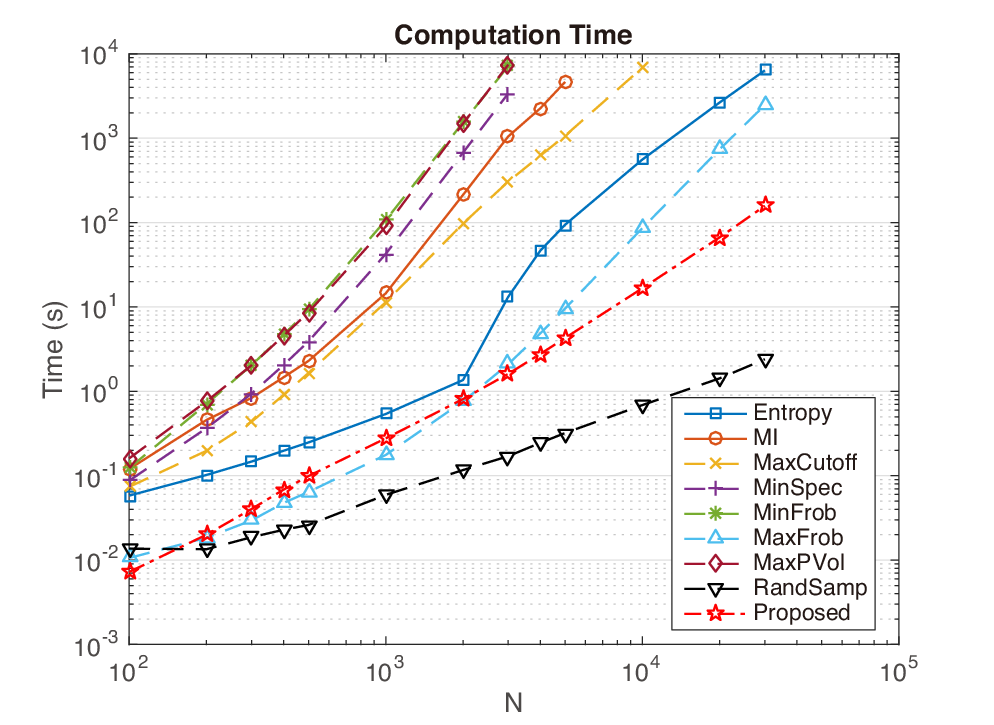}
\caption{Execution time comparison for random sensor graph. Note that both axes are represented on logarithmic scales.}
\label{fig:time}
\end{figure}

\begin{table}[t]
    \centering
    \caption{Speedup Factor of our Method with respect to Alternative Approaches for $N = 2000$}
    \label{tab:speedup_factor}
    \begin{tabular}{c|r}
    \hline
    Methods     &  \multicolumn{1}{c}{Speedup factor}\\\hline
    Entropy & $1.68$\\
    MI & $262.85$\\
    MaxCutoff & $120.24$\\
    MinSpec & $833.01$\\
    MinFrob & $1898.37$\\
    MaxFrob & $0.94$\\
    MaxPVol & $1843.99$\\
    RandSamp & $0.14$\\\hline
    \end{tabular}
\end{table}

The proposed method uses the kernel $g(\lambda)=\exp(-s\lambda)$ with $s = \nu p_e p_s p_f/\lambda_{\max}$, where $\nu \in \mathbb{R}_+$ is a parameter, $p_e := |\mathcal{E}|/N$ is the edge probability, $p_s := |\mathcal{S}|/N$ is the sampling ratio, and $p_f := |\mathcal{F}|/N$ is the (normalized) bandwidth. For all the experiments, $\nu$ was experimentally set to $\nu=220$.\footnote{The optimal $\nu$ differs for different graphs, but this $\nu$ works well for our experiments. The automatic parameter setting will be an interesting topic in the future.} During the selection process, $g(\lambda)$ is approximated with Chebyshev polynomial approximation with the order $P = 12$. For the signal prediction, we fixed $k=12$.

\begin{figure*}[tp]
\centering
 \subfigure[][Original graph]{\includegraphics[width=.2\linewidth]{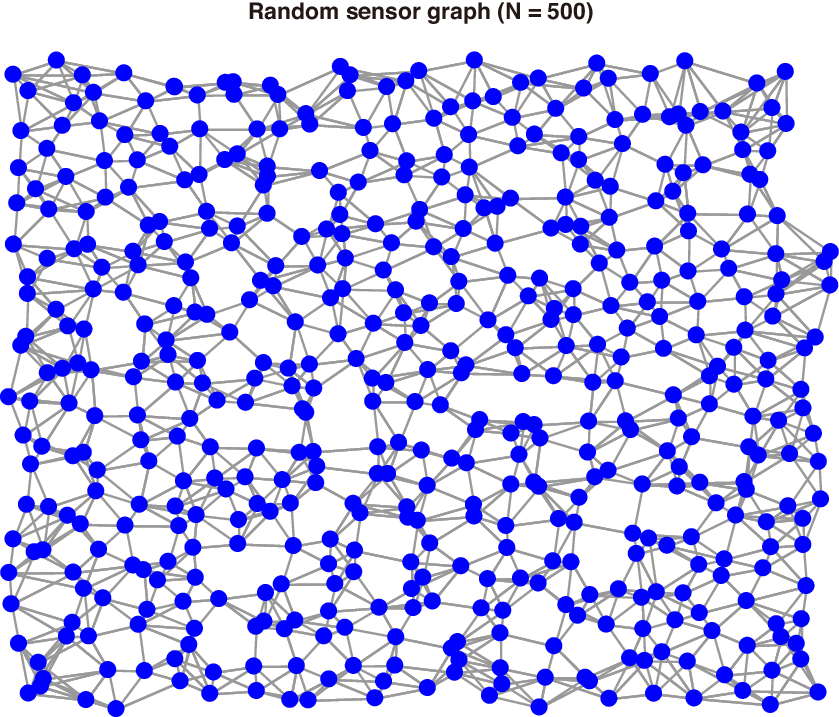}}
 \subfigure[][Entropy]{\includegraphics[width=.2\linewidth]{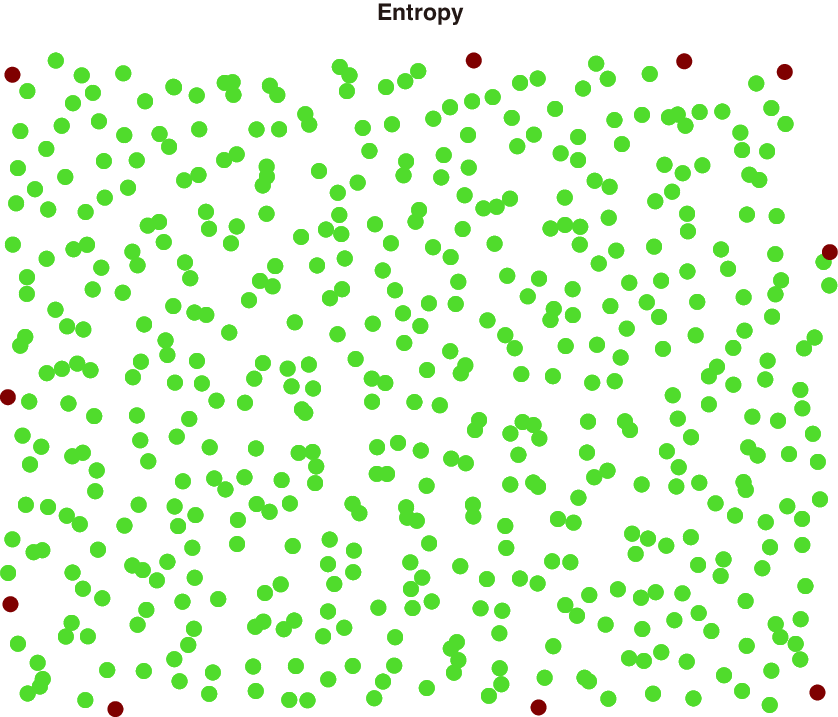}}
 \subfigure[][MI]{\includegraphics[width=.2\linewidth]{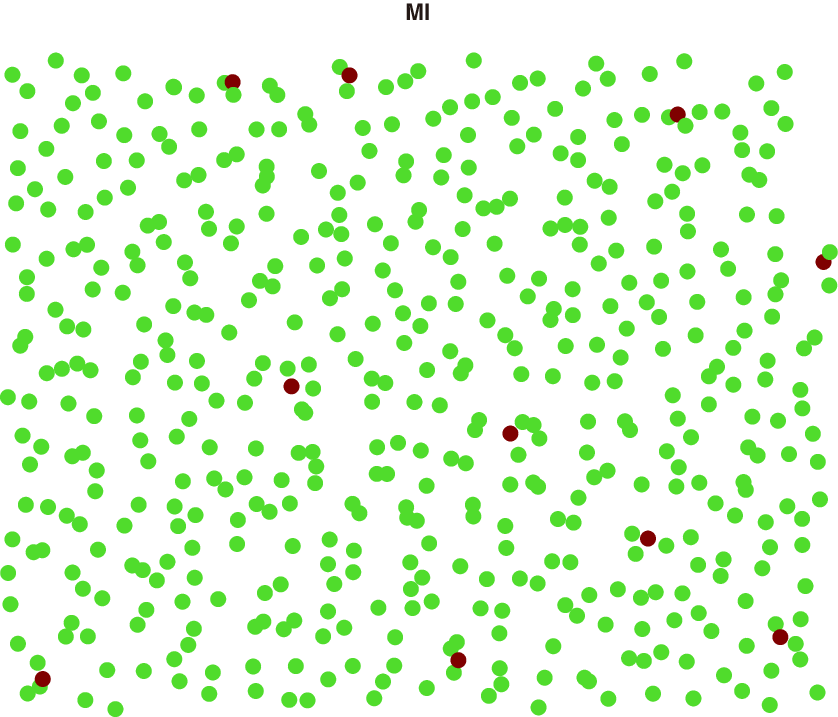}}\\
 \subfigure[][MaxCutoff]{\includegraphics[width=.2\linewidth]{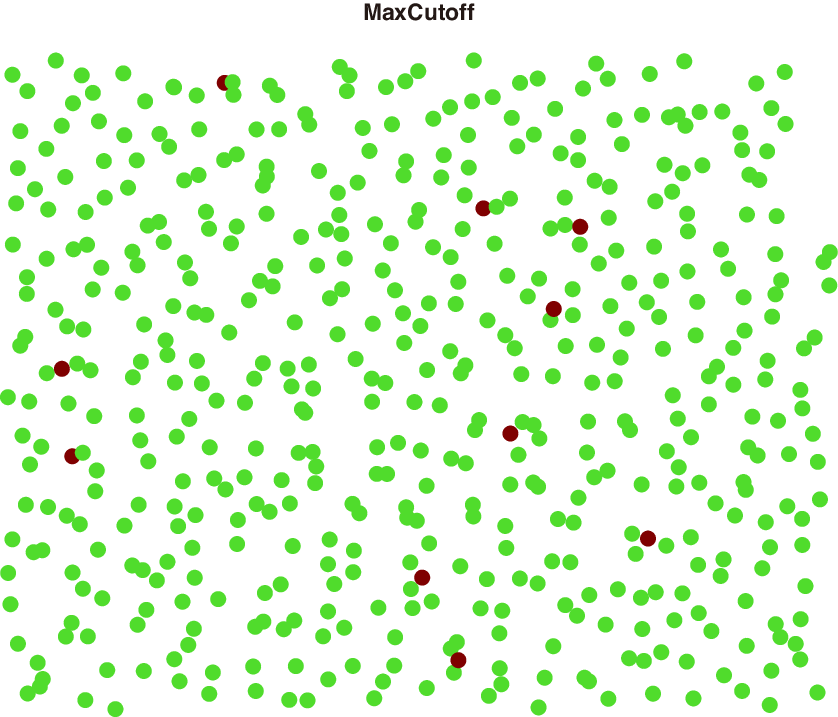}}
 \subfigure[][MinSpec]{\includegraphics[width=.2\linewidth]{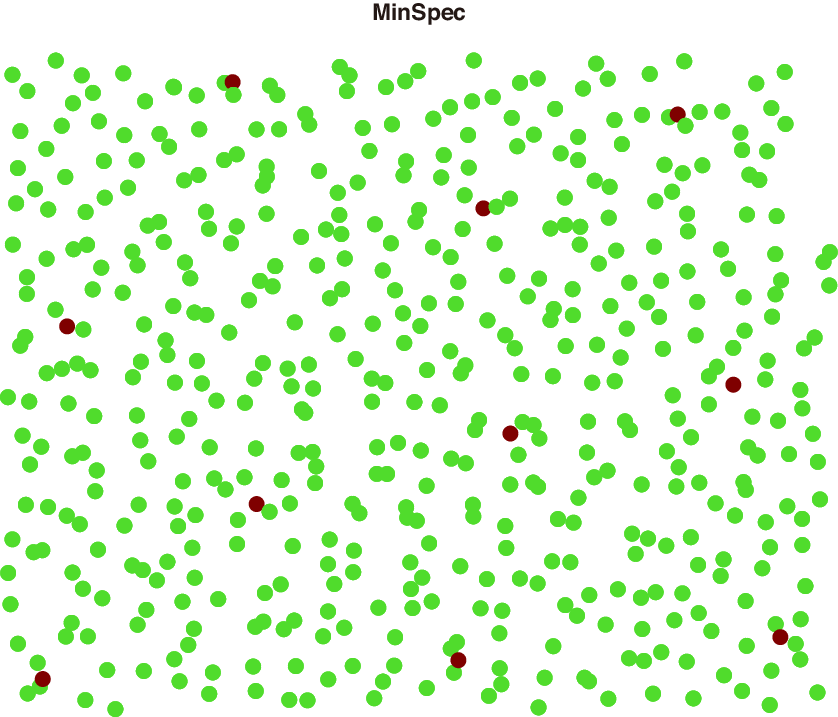}}
 \subfigure[][MinFrob]{\includegraphics[width=.2\linewidth]{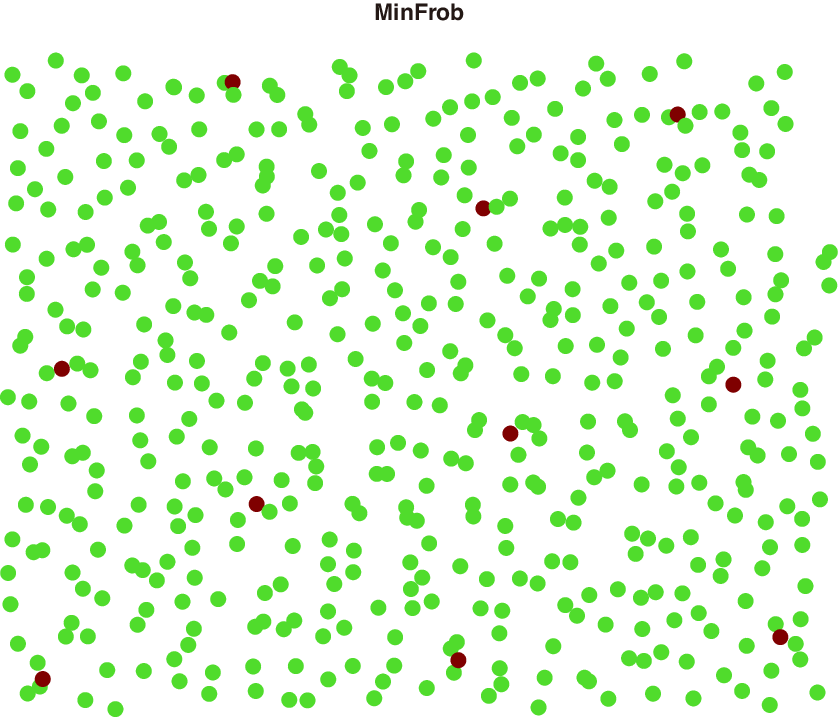}}
 \subfigure[][MaxFrob]{\includegraphics[width=.2\linewidth]{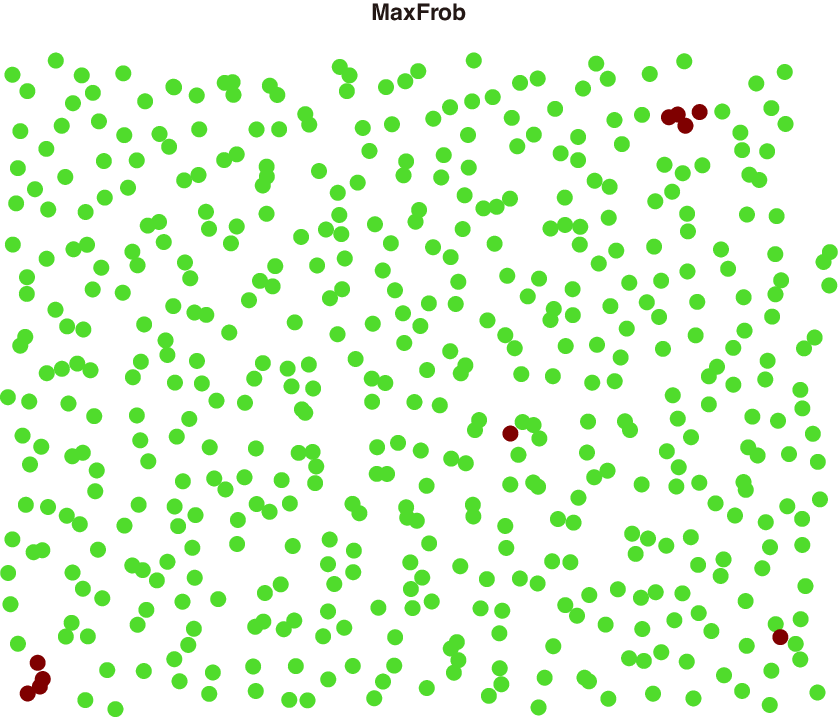}}\\
 \subfigure[][MaxPVol]{\includegraphics[width=.2\linewidth]{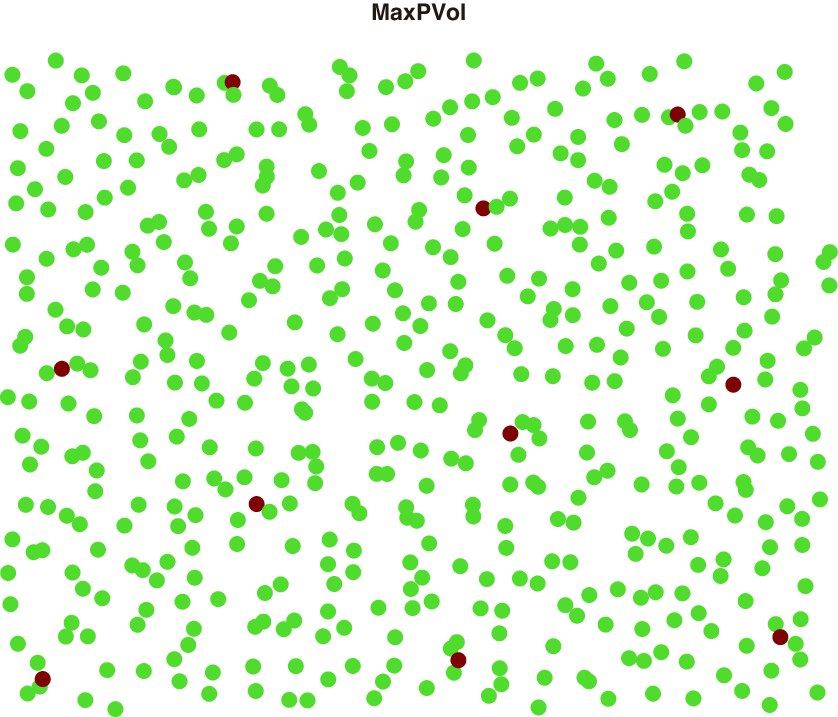}}
 \subfigure[][RandSamp (realization \#1)]{\includegraphics[width=.2\linewidth]{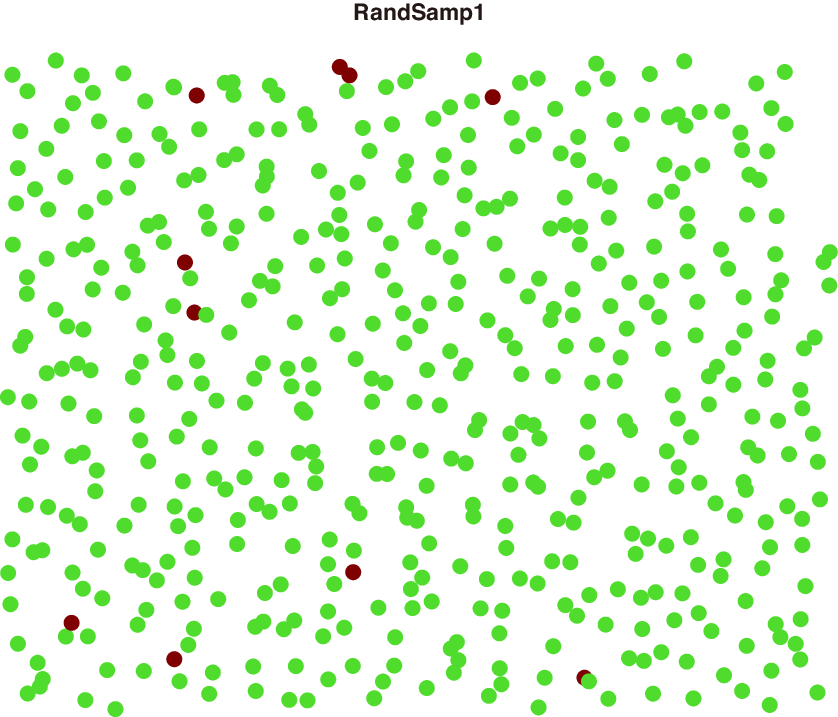}}
 \subfigure[][RandSamp (realization \#2)]{\includegraphics[width=.2\linewidth]{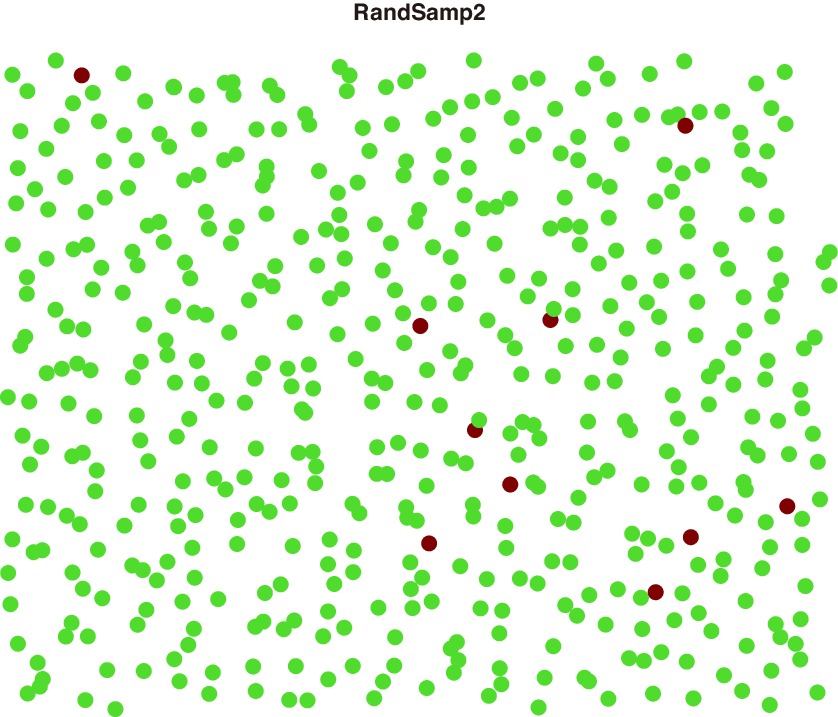}}
 \subfigure[][Proposed]{\includegraphics[width=.24\linewidth]{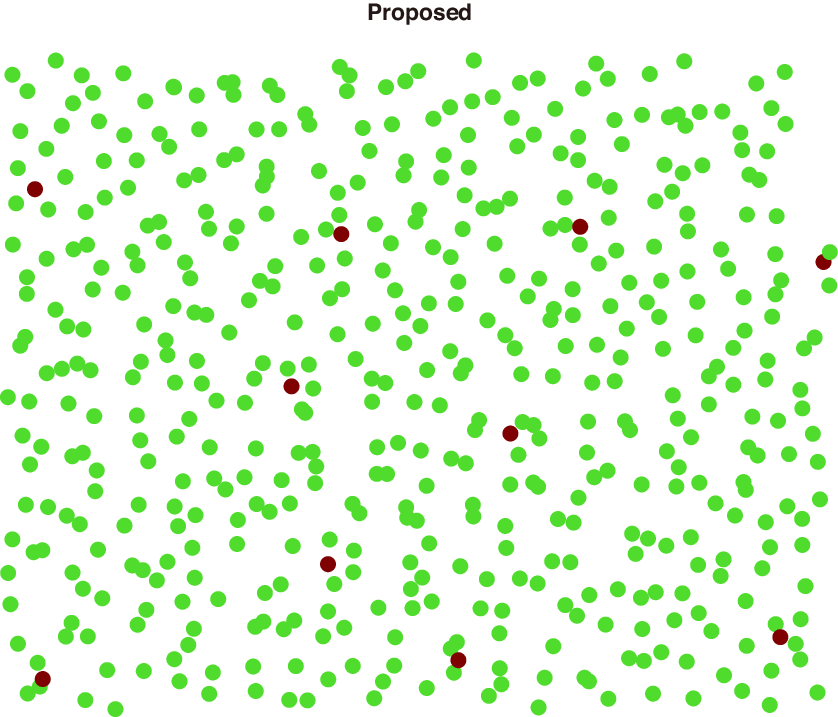}}
\caption{Selected vertices for a random sensor graph ($N = 500$). Ten vertices are selected, which are colored in red.}
\label{fig:node_sel}
\end{figure*}

All the experiments were performed in MATLAB R2017a, running on a PC with an Intel Xeon E5 3 GHz CPU and 64 GB RAM. The MATLAB toolbox for submodular function optimization \cite{Krause2010, Krause2007} was used for implementations of the entropy and MI criteria.

\subsection{Execution Time}
First, we compare the execution time for choosing $|\mathcal{S}|= N/10$ vertices for various $N$. Figure \ref{fig:time} shows the execution time comparison plotted against $N$ for the random sensor graph. The results are given by the average of $10$ independent runs.

Among the deterministic approaches, Entropy, MaxFrob, and the proposed method are faster than the other methods. Specifically, for $N= 2000$ (thus $|\mathcal{S}| = 200$), the speedup factor of our method with respect to an alternative method, i.e.,
\[
    \frac{\text{Comp. time of alternative method}}{\text{Comp. time of proposed method}}
\]
is summarized in Table \ref{tab:speedup_factor}. The proposed method is $>100$ times faster than the methods with high prediction accuracies (presented in the following sections): MI, MaxCutoff, MinSpec, MinFrob, and MaxPVol.

Although Entropy and MaxFrob are very fast compared to the other conventional methods and their computation times are comparable to that of the proposed method, their performances on SSS and the signal value prediction are significantly worse than the other methods. RandSamp is significantly faster than the other methods, including the proposed method, because it is a ``one-shot'' algorithm. However, on average, the prediction performance of the proposed method outperforms that of RandSamp. These results are further discussed in the next subsections.

\begin{figure*}[tp]
\centering
 \subfigure[][Random Sensor]{\includegraphics[width=.32\linewidth]{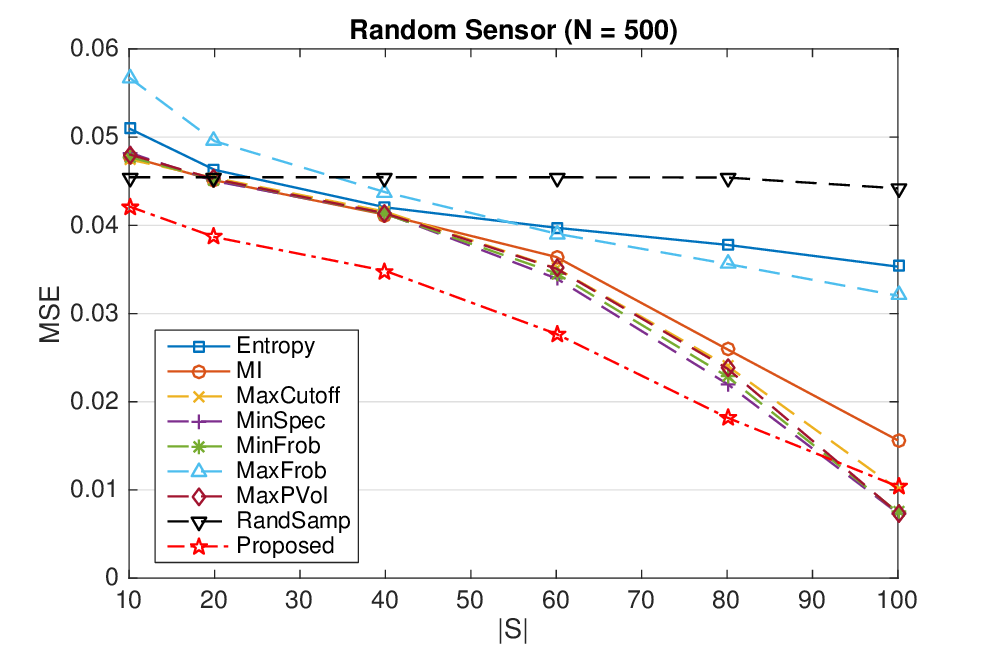}}
 \subfigure[][Erd\H{o}s--R\'{e}nyi]{\includegraphics[width=.32\linewidth]{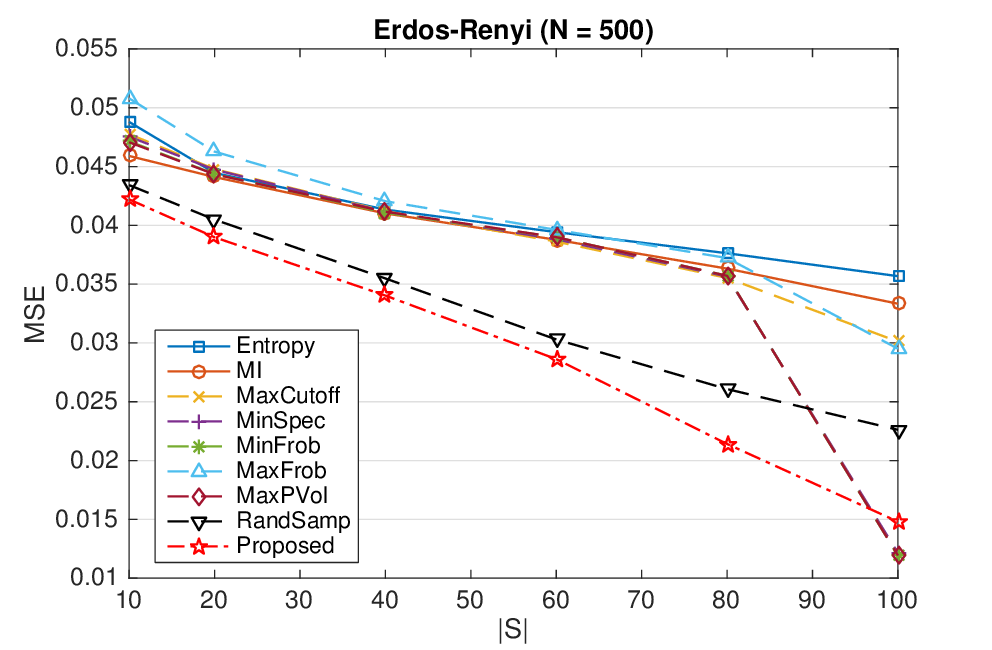}}
 \subfigure[][Random Regular]{\includegraphics[width=.32\linewidth]{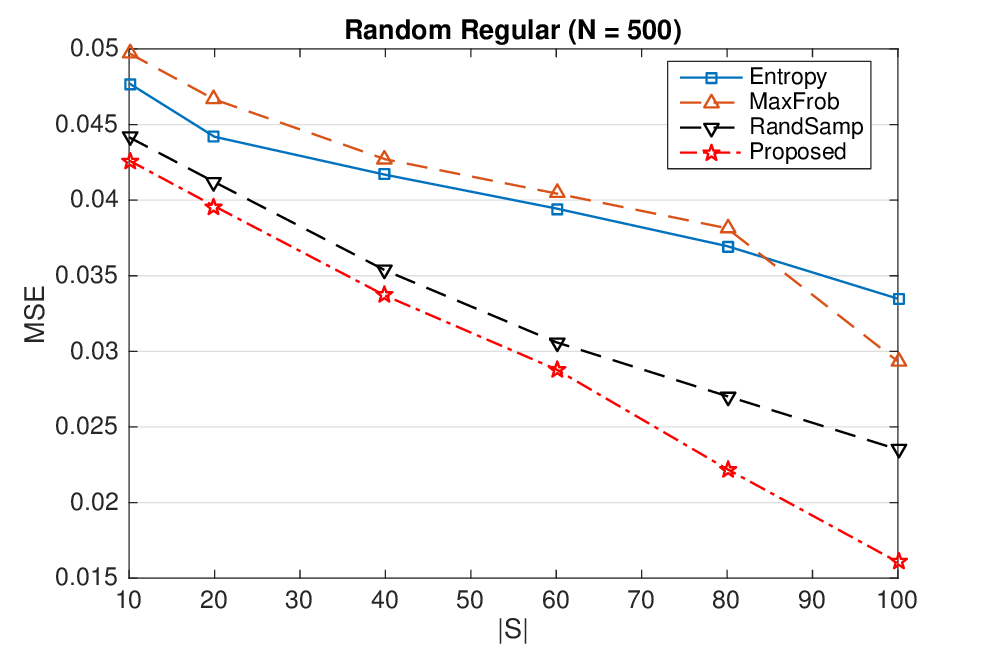}}\\
 \subfigure[][Barab\'{a}si--Albert]{\includegraphics[width=.32\linewidth]{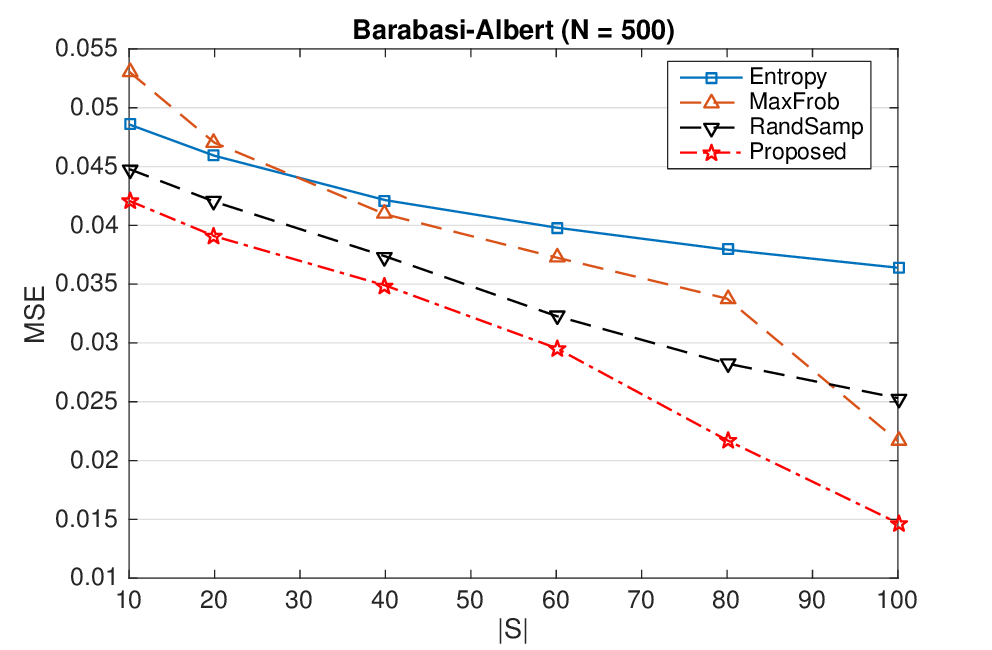}}
 \subfigure[][Community]{\includegraphics[width=.32\linewidth]{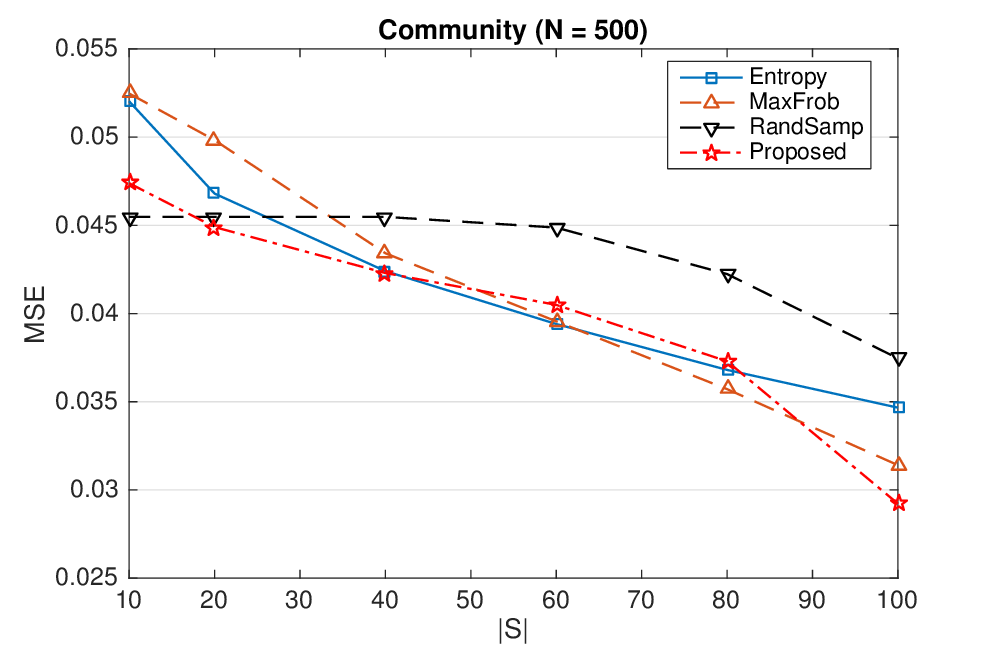}}
 \subfigure[][Minnesota Traffic]{\includegraphics[width=.32\linewidth]{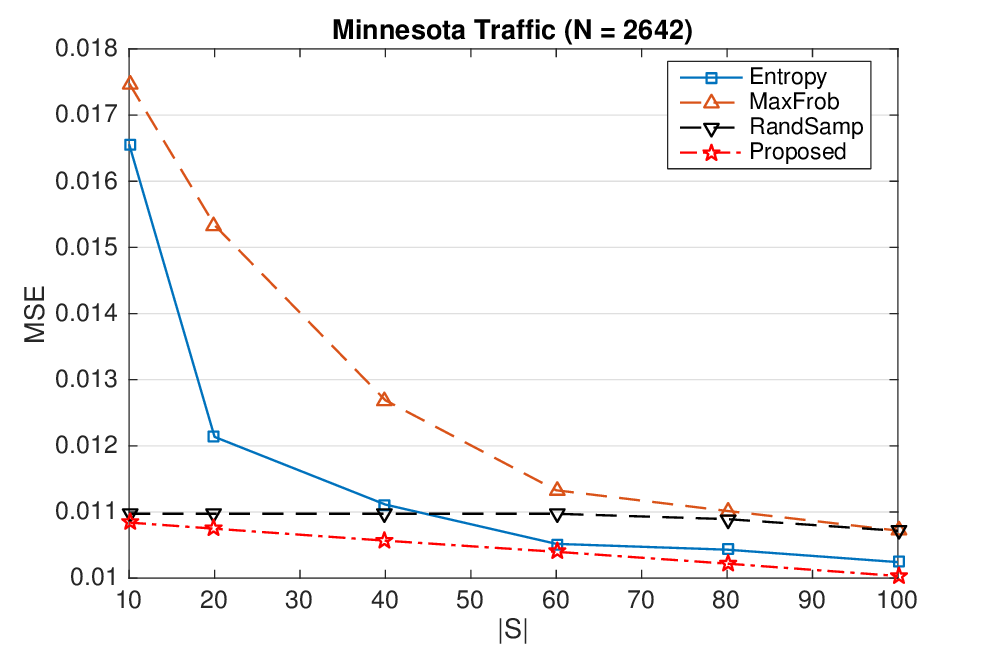}}
\caption{MSE comparison of estimated signals (average of $100$ tested signals).}
\label{fig:MSE}
\end{figure*}

\subsection{Comparison of Selected Vertices}
We show the vertices selected using the conventional and proposed approaches.
The original graph is a random sensor graph with $N=500$, and we select $|\mathcal{S}|=10$ vertices. The underlying graph and sampling results are shown in Fig. \ref{fig:node_sel}. Because RandSamp selects different vertices in every selection process, we show two sampling realizations from the same sampling distribution probability.

It can be observed that MI, MaxCutoff, MinSpec, MinFrob, MaxPVol, and the proposed method select evenly distributed sampling positions (and the proposed method is the fastest among the six). Entropy selects many vertices at the corners or boundaries of the graph because they have fewer connections with other vertices.
Because MaxFrob does not take into account the position of the already-selected vertices in each iteration, it often selects a vertex very close to an already selected one, as can be seen in Fig. \ref{fig:node_sel}(g) (some of them are almost overlapped). It therefore leads to large reconstruction errors, as shown in the next subsection.
Two sampling sets of RandSamp are quite different from each other, and sometimes vertices very close to each other are selected, as shown in Fig. \ref{fig:node_sel}(i).

\subsection{Comparison of Reconstruction Errors}
We compare the prediction errors between the proposed and conventional SSS methods.
Owing to the long execution times, all the SSS methods have been compared only for random sensor and ER graphs. For the other graphs, we compare the performances of the relatively fast methods: Entropy, MaxFrob, RandSamp, and the proposed method. The results are the average of $100$ runs.

The GP-based methods, MaxCutoff \cite{Anis2016} and MinSpec \cite{Chen2015} use the reconstruction shown in \eqref{rec} where $\mathcal{F}$ is the set of Laplacian eigenvalues less than $\lambda_{F}$ (for MinSpec \cite{Chen2015}) or the estimated cutoff frequency $\Omega_k(\mathcal{S})$ (for the other methods). For MaxCutoff, $k = 14$ is used because it has been effective for the reconstruction of noisy bandlimited signals \cite{Anis2016}. For GP-based methods, $k = 6$ is used because it presents better performance in our experiments.
MinFrob, MaxFrob, and MaxPVol \cite{Tsitsv2016} use the reconstruction shown in \eqref{rec_tes}, and the proposed method uses the reconstruction shown in \eqref{rec}.
RandSamp reconstructs the signal with the same method as the original paper \cite{Puy2018}: The quadratic equation with the Laplacian smoothness regularizer term. The regularizer function $g_{\text{reg}}(\la)$ we used is that proposed in \cite{Puy2018}, where $g_{\text{reg}}(\la) = \la^4$.

The average MSEs between the predicted and original signals are summarized in Fig. \ref{fig:MSE}. It is clear that the proposed method presents the lowest MSEs for almost all cases, with the exception of the community graph. The MSEs of RandSamp greatly depend on the specific graphs: It is good for the ER, random regular, and BA graphs, whereas its MSEs are not improved for the random sensor, community, and Minnesota Traffic graphs, even when we select a large number of vertices. This could be due to the reconstruction parameter $\gamma$. Graph sampling theory-based approaches perform well when enough samples, i.e., close to the cutoff frequency, are selected. In contrast, they have larger MSEs than the proposed method for small $|\mathcal{S}|$. 

\section{Conclusion}
\label{sec:con}
We proposed a SSS method based on the localization operator for graph signals. The proposed method has strong connections with the conventional GP-based sensor selection and the graph frequency-based SSS. The proposed SSS does not need the eigendecomposition of the graph Laplacian, whereas it still considers the graph frequency information as well as vertex information. It is significantly faster than the existing approaches, and its performance is better than those that are proven through numerical experiments.

\section*{Appendices}
\label{sec:con}
\subsection{Derivation of Conventional Approaches: Sensor Position Selection}
\subsubsection{Entropy}
In \eqref{eq:en}, the sensors are selected such that the uncertainty of a measurement with respect to previous measurements is maximized\cite{Cressi2015, Shewry1987,Sharma2015}:
\begin{equation}
\begin{split}
\mathcal{S}^*&=\argmin_{\mathcal{S}\subset\mathcal{V}:|\mathcal{S}|=F}H(\bm{f}_{\mathcal{S}^c}|\bm{f}_{\mathcal{S}})\\
&=\argmax_{\mathcal{S}\subset\mathcal{V}:|\mathcal{S}|=F}H(\bm{f}_{\mathcal{S}})=\argmax_{\mathcal{S}\subset\mathcal{V}:|\mathcal{S}|=F}\log\det{[\mathbf{ K}_{\mathcal{S}}]},
\end{split}
\end{equation}
where $H(\cdot)$ is (conditional) entropy. 

Because the problem in \eqref{eq:en} is NP-complete, a greedy algorithm has been proposed in \cite{Cressi2015, Shewry1987}. We first set $\mathcal{S}=\emptyset$ and add a sensor, which ensures the maximum increase in the uncertainty of the observed sensors, to $\mathcal{S}$ from the set of unselected sensors $\mathcal{S}^c$ one by one. 

The entropy of the random variable ${f}({y})$, where $y$ is the sensor of interest, conditioned on the variable $\bm{f}_{\mathcal{S}}$ is a monotonic function of its variance: 
\begin{equation}
H({f}({y})|\bm{f}_{\mathcal{S}})=\frac{1}{2}\log(2\pi e (\mathcal{K}(y,y)-\mathbf{ K}_{y\mathcal{S}}\mathbf{ K}_{\mathcal{S}}^{-1}\mathbf{ K}_{\mathcal{S}y})).
\label{conen}
\end{equation}
Hence, the vertex that satisfies \eqref{eq:en3} is selected at each step.

\subsubsection{MI}
\eqref{eq:mi1} maximizes the MI between the selected locations $\mathcal{S}$ and unselected locations $\mathcal{S}^c$, i.e., it selects the locations that more significantly reduce the uncertainty of the rest of the space\cite{Krause2008,Sharma2015}:
\begin{equation}
\begin{split}
\mathcal{S}^*&=\argmax_{\mathcal{S}\subset\mathcal{V}:|\mathcal{S}|=F}H(\bm{f}_{\mathcal{S}^c})-H(\bm{f}_{\mathcal{S}^c}|\bm{f}_{\mathcal{S}})\\
&=\argmax_{\mathcal{S}\subset\mathcal{V}:|\mathcal{S}|=F}\log\det{[\mathbf{ K}_{\mathcal{S}}]}+\log\det{[\mathbf{ K}_{\mathcal{S}^c}]}.
\label{eq:mi1_2}
\end{split}
\end{equation}
From \eqref{conen} and \eqref{eq:mi1_2}, a greedy method\cite{Krause2008} that adds sensor $y^*$ satisfying \eqref{eq:mi3} is used for the optimization.

\subsection{Derivation of Conventional Approaches: Graph Sampling Based on Graph Fourier Basis}

\subsubsection{Based on Cutoff Frequency (MaxCutoff)}
\cite{Anis2016} introduces a measure of quality for the sampling sets, namely the cutoff frequency, and selects the sampled vertices to maximize the cutoff frequency. It guarantees unique reconstruction and does not need the calculation of the graph Fourier basis.

The cutoff frequency associated with the subset $\mathcal{S}$ is a bound on the maximum frequency of a signal that can be perfectly recovered from the samples on the subset $\mathcal{S}$. Let us denote by $PW_{\omega}(\mathcal{G})\in\mathbb{R}^N$ the \textit{Paley--Wiener space} which is the space of all $\omega$-bandlimited signals, by $L_2(\mathcal{S}^c)$ the space of signals with zero values on $\mathcal{S}$, i.e., if $\bm{\phi}\in L_2(\mathcal{S}^c)$ then $\bm{\phi}=[\bm{f}_{\mathcal{S}^c}^T\ \mathbf{0}^T]^T$, and by $\omega(\bm{\phi})$ the minimum eigenvalue of $\bm{\phi}$ that have non-zero graph Fourier coefficients. 
\cite{Anis2014} states the sampling theorem for graph signals as follows.
\begin{theorem}[Graph Sampling Theorem{\cite[Theorem 2]{Anis2014}}]
The signal on a graph can be perfectly reconstructed from signal values $\bm{f}_{\mathcal{S}}$ on $\mathcal{S}$
if and only if $\bm{f}\in PW_{\omega}(\mathcal{G})$, where 
\begin{equation}
\omega<\omega_c(\mathcal{S}):=\inf_{\bm{\phi}\in L_2(\mathcal{S}^c)}\omega(\bm{\phi}),
\end{equation}
and $\omega_c(\mathcal{S})$ is the exact cutoff frequency.
\end{theorem}

To avoid the calculation of the true cutoff frequency $\omega_c(\mathcal{S})$, which needs the computation of the graph Fourier basis, we can use the estimated cutoff frequency $\Omega_k(\mathcal{S})$ for the sampling set $\mathcal{S}$:
\begin{equation}
\Omega_k(\mathcal{S})=\inf_{\bm{\phi}\in L_2(\mathcal{S}^c)} \left(\frac{\bm{\phi}^T{\mathbf{ L}}^k\bm{\phi}}{\bm{\phi}^T\bm{\phi}}\right),
\label{eq:estcf}
\end{equation}
where $k\in\mathbb{Z}^{+}$ is a parameter that provides a trade-off between performance and complexity. A large $k$ leads the estimated cutoff frequency to be close to the actual bandwidth. 
\eqref{gs1} selects vertices so as to maximize the estimated cutoff frequency in \eqref{eq:estcf}:
\begin{equation}
\begin{split}
{\mathcal{S}}^*=\argmax_{\mathcal{S}\subset\mathcal{V}:|\mathcal{S}|=F} \Omega_k(\mathcal{S})&=\argmax_{\mathcal{S}\subset\mathcal{V}:|\mathcal{S}|=F}\min_{\bm{\psi}}\frac{{\bm{\psi}}^T(\mathbf{ L}^k)_{\mathcal{S}^c}{\bm{\psi}}}{{\bm{\psi}}^T{\bm{\psi}}}\\
&=\argmax_{\mathcal{S}\subset\mathcal{V}:|\mathcal{S}|=F}\mu_\text{min}((\mathbf{ L}^k)_{\mathcal{S}^c}).
\end{split}
\end{equation}

\subsubsection{Based on Error Minimization}
The objective functions shown in \eqref{gs2-1} and \eqref{gs2-2} are determined so as to minimize the reconstruction error. It assumes that the measured signal is corrupted by noise and/or not bandlimited; that is, the measured signal is defined as $\bm{o}=\bm{f}_{\mathcal{S}}+\bm{n}_{\mathcal{S}}$. 
This method also uses the reconstruction method in \eqref{rec}; then, the reconstruction error becomes $\bm{e}=\bm{f}-\mathbf{ U}_{\mathcal{V}\mathcal{F}}(\mathbf{ U}_{\mathcal{S}\mathcal{F}})^+(\bm{f}_{\mathcal{S}}+\bm{n}_{\mathcal{S}})=\mathbf{ U}_{\mathcal{V}\mathcal{F}}(\mathbf{ U}_{\mathcal{S}\mathcal{F}})^+\bm{n}_{\mathcal{S}}$. 
\begin{itemize}
\renewcommand{\labelenumi}{(\roman{enumi})}
\item\textit{MinSpec:} 
\eqref{gs2-1} is obtained by minimizing the $\ell_2$ norm of the error:
\begin{equation}
\begin{split}
{\mathcal{S}}^*&=\argmin_{\mathcal{S}\subset\mathcal{V}:|\mathcal{S}|=F}\|\bm{e}\|_2\\
&=\argmin_{\mathcal{S}\subset\mathcal{V}:|\mathcal{S}|=F}\|\mathbf{ U}_{\mathcal{V}\mathcal{F}}(\mathbf{ U}_{\mathcal{S}\mathcal{F}})^+\bm{n}_{\mathcal{S}}\|_2\\
&\leq\argmin_{\mathcal{S}\subset\mathcal{V}:|\mathcal{S}|=F}\|\mathbf{ U}_{\mathcal{V}\mathcal{F}}\|_2\|(\mathbf{ U}_{\mathcal{S}\mathcal{F}})^+\|_2\|\bm{n}_{\mathcal{S}}\|_2\\
&=\argmin_{\mathcal{S}\subset\mathcal{V}:|\mathcal{S}|=F}\|(\mathbf{ U}_{\mathcal{S}\mathcal{F}})^+\|_2\\
&=\argmax_{\mathcal{S}\subset\mathcal{V}:|\mathcal{S}|=F} \sigma_{\text{min}}(\mathbf{ U}_{\mathcal{S}\mathcal{F}}).
\end{split}
\end{equation}
\item\textit{MinTrac:}  \eqref{gs2-2} minimizes the mean squared errors, i.e., minimizes the trace of the error covariance matrix $\mathbf{ E}$:
\begin{equation}
\begin{split}
\mathcal{S}^*=&\argmin_{\mathcal{S}\subset\mathcal{V}:|\mathcal{S}|=F} \text{tr}[\mathbf{ E}]\\
:=&\argmin_{\mathcal{S}\subset\mathcal{V}:|\mathcal{S}|=F} \text{tr}[\bm{e}\bm{e}^*]\\
=&\argmin_{\mathcal{S}\subset\mathcal{V}:|\mathcal{S}|=F} \mathbf{ U}_{\mathcal{V}\mathcal{F}}(\mathbf{ U}_{\mathcal{S}\mathcal{F}}^*\mathbf{ U}_{\mathcal{S}\mathcal{F}})^{-1}\mathbf{ U}_{\mathcal{V}\mathcal{F}}\\
=&\argmin_{\mathcal{S}\subset\mathcal{V}:|\mathcal{S}|=F}\text{tr}[(\mathbf{ U}_{\mathcal{S}\mathcal{F}}^*\mathbf{ U}_{\mathcal{S}\mathcal{F}})^{-1}].
\end{split}
\end{equation}
\end{itemize}

\subsubsection{Based on Localized Basis}
\cite{Tsitsv2016} uses the basis localized both in the vertex and graph frequency domains for graph sampling and signal recovery. The basis $\bm{\psi}_i$, $i=0,1,\ldots,N-1$ perfectly localized in the graph frequency domain and highly localized in the vertex domain, is designed by solving the following problem:
\begin{equation}
\begin{split}
\bm{\psi}_i=\argmax_{\bm{\psi}_i}\|\mathbf{ D}_\text{ver}\bm{\psi}_i\|_2\\
\text{s. t. } \|\bm{\psi}_i\|_2=1, \mathbf{ D}_\text{sp}\bm{\psi}_i=\bm{\psi}_i,\\
\langle\bm{\psi}_i, \bm{\psi}_j\rangle=0, j=1,\ldots,i-1.
\end{split}
\end{equation}
Its optimal solution coincides with the eigenvectors of $\mathbf{ D}_\text{sp}\mathbf{ D}_\text{ver}\mathbf{ D}_\text{sp}$, i.e., $\bm{\psi}_i=\bm{v}_i(\mathbf{ D}_\text{sp}\mathbf{ D}_\text{ver}\mathbf{ D}_\text{sp})$. 

The reconstruction shown in \eqref{rec_tes} indicates that the sampled signal $\mathbf{ D}_\text{ver}\bm{f}$ is interpolated by $\bm{\psi}_i$ for recovering the original signal:
\begin{equation}
\begin{split}
\widehat{\bm{f}}&=\sum^{|\mathcal{F}|-1}_{i=0}\frac{1}{\sigma_i^2(\mathbf{ D}_\text{sp}\mathbf{ D}_\text{ver}\mathbf{ D}_\text{sp})}\langle\mathbf{ D}_\text{ver}\bm{f},\bm{\psi}_i\rangle\bm{\psi}_i
\\&=\bm{\Psi}_{\mathcal{V},\mathcal{F}}\bm{\Sigma}^{-1}_{\mathcal{F},\mathcal{F}}\bm{\Psi}^*_{\mathcal{V},\mathcal{F}}\mathbf{ D}_\text{ver}\bm{f}
\\&=(\mathbf{ D}_\text{sp}\mathbf{ D}_\text{ver}\mathbf{ D}_\text{sp})^+\mathbf{ D}_\text{ver}\bm{f}.
\end{split}
\end{equation}

\eqref{gs3-1}--\eqref{gs3-3} reduce the error caused by noise which can be written as $\bm{e}=\bm{f}-\widetilde{\bm{o}}=\bm{f}-(\mathbf{ D}_\text{sp}\mathbf{ D}_\text{ver}\mathbf{ D}_\text{sp})^+\mathbf{ D}_\text{ver}(\bm{f}+\bm{n})
=(\mathbf{ D}_\text{sp}\mathbf{ D}_\text{ver}\mathbf{ D}_\text{sp})^+\mathbf{ D}_\text{ver}\bm{n}$ where $\widetilde{\bm{o}}$ is the sampled signal with zero interpolation: $\widetilde{\bm{o}}_\mathcal{S}=\bm{o}_\mathcal{S}$ and  $\widetilde{\bm{o}}_{\mathcal{S}^c}=\bm{0}_{|\mathcal{S}^c|}$.
\begin{itemize}
\renewcommand{\labelenumi}{(\roman{enumi})}
\item\textit{MinFrob:}
\eqref{gs3-1} is obtained by minimizing the Frobenius norm of the error $(\mathbf{ D}_\text{sp}\mathbf{ D}_\text{ver}\mathbf{ D}_\text{sp})^+$:
\begin{equation}
\begin{split}
\mathcal{S}^*&=\argmin_{\mathcal{S}\subset\mathcal{V}:|\mathcal{S}|=F} \|(\mathbf{ D}_\text{sp}\mathbf{ D}_\text{ver}\mathbf{ D}_\text{sp})^+\mathbf{ D}_\text{ver}\bm{n}\|_F\\
&\leq\argmin_{\mathcal{S}\subset\mathcal{V}:|\mathcal{S}|=F} \|(\text{diag}(\bm{1}_{\mathcal{F}})\mathbf{ U}^*\mathbf{ D}_\text{ver})^+\|_F\|\mathbf{ D}_\text{ver}\bm{n}\|_F.\\
&=\argmin_{\mathcal{S}\subset\mathcal{V}:|\mathcal{S}|=F}\sum_{i=0}^{|\mathcal{F}|-1}\frac{1}{\lambda_i(\mathbf{ U}_\mathcal{SF}^*)}.
\end{split}
\end{equation}
\item\textit{MaxFrob:}
\eqref{gs3-2} is the approximation of (i) and maximizes the Frobenius norm of $\mathbf{ D}_\text{sp}\mathbf{ D}_\text{ver}\mathbf{ D}_\text{sp}$. 
It does not need any eigendecomposition of the variation operator for solving the problem. 
\item\textit{MaxPVol:} 
\eqref{gs3-3} maximizes the volume of the parallelepiped formed with the columns of $\mathbf{ U}_{{\mathcal{S}}\mathcal{F}}^*$ which can be computed by the determinant of  $\mathbf{ U}_{\mathcal{SF}}^*\mathbf{ U}_{\mathcal{SF}}$.
\end{itemize}

\subsection{Derivations of Objective Functions in IV-B-1}
The spectral norm of the error covariance matrix can be bounded as
\begin{equation}
\begin{split}
\|\mathbf{ E}\|_2&=\|\mathbf{ T}^{k/2}((\mathbf{ T}^{k/2})_{{\mathcal S}{\mathcal V}}^*(\mathbf{ T}^{k/2})_{{\mathcal S}{\mathcal V}})^+\mathbf{ T}^{k/2}\|_2\\
&\leq\|\mathbf{ T}^{k/2}\|_2\|((\mathbf{ T}^{k/2})_{{\mathcal S}{\mathcal V}})^+((\mathbf{ T}^{k/2})_{{\mathcal S}{\mathcal V}})^+)^*\|_2\|\mathbf{ T}^{k/2}\|_2.\\
\end{split}
\end{equation}
Therefore, the objective function for minimizing the error covariance matrix can be represented as follows:
\begin{align}
\mathcal{S}^*&=\argmin_{\mathcal{S}\subset\mathcal{V}:|\mathcal{S}|=F}\|\mathbf{ E}\|_2\nonumber\\
&=\argmin_{\mathcal{S}\subset\mathcal{V}:|\mathcal{S}|=F}\|((\mathbf{ T}^{k/2})_{{\mathcal S}{\mathcal V}})^+\|_2\label{minspec}\\
&\leq\argmin_{\mathcal{S}\subset\mathcal{V}:|\mathcal{S}|=F}\|((\mathbf{ T}^{k/2})_{{\mathcal S}{\mathcal V}})^+\|_F\label{minFrob}\\
&=\argmin_{\mathcal{S}\subset\mathcal{V}:|\mathcal{S}|=F}\text{tr}[((\mathbf{ T}^{k/2})_{{\mathcal S}{\mathcal V}}(\mathbf{ T}^{k/2})_{{\mathcal S}{\mathcal V}}^*)^+]\nonumber\\
&=\argmin_{\mathcal{S}\subset\mathcal{V}:|\mathcal{S}|=F}\text{tr}[((\mathbf{ T}^{k})_{{\mathcal S}})^{-1}]\label{mintr}\\
&\approx\argmax_{\mathcal{S}\subset\mathcal{V}:|\mathcal{S}|=F}\text{tr}[(\mathbf{ T}^{k})_{{\mathcal S}}]\nonumber.
\end{align}
The above functions minimize the spectral norm \eqref{minspec}, Frobenius norm \eqref{minFrob}, or trace \eqref{mintr} of the error covariance matrix.  Furthermore, because $(g(\mathbf{ \Lambda}_{{\mathcal F}^c})/\beta)^k\approx \mathbf{ 0}_{|{\mathcal F}^c|}$ for large $k$, the determinant of the error covariance matrix becomes
\begin{equation}
\begin{split}
&\det[\mathbf{ E}]\\
=&\det[\mathbf{ T}^{k/2}(\mathbf{ U}g(\mathbf{ \Lambda})^{k/2}\mathbf{ U}_{{\mathcal S}{\mathcal V}}^*\mathbf{ U}_{{\mathcal S}{\mathcal V}}g(\mathbf{ \Lambda})^{k/2}\mathbf{ U}^*)^+\mathbf{ T}^{k/2}]\\
=&\det[\mathbf{ T}^{k/2}\mathbf{ U}(g(\mathbf{ \Lambda}_{\mathcal F})^{k/2}\mathbf{ U}_{{\mathcal S}{\mathcal F}}^*\mathbf{ U}_{{\mathcal S}{\mathcal F}}g(\mathbf{ \Lambda}_{\mathcal F})^{k/2})^+\mathbf{ U}^*\mathbf{ T}^{k/2}]\\
=&\det[\mathbf{ T}^{k/2}\mathbf{ U}]\det[(g(\mathbf{ \Lambda}_{\mathcal F})^{k/2}\mathbf{ U}_{{\mathcal S}{\mathcal F}}^*\mathbf{ U}_{{\mathcal S}{\mathcal F}}g(\mathbf{ \Lambda}_{\mathcal F})^{k/2})^+]\\
&\times\det[\mathbf{ U}^*\mathbf{ T}^{k/2}]\\
=&\det[\mathbf{ T}^{k/2}\mathbf{ U}]\det[(\mathbf{ U}_{{\mathcal S}{\mathcal F}}g(\mathbf{ \Lambda}_{\mathcal F})^{k}\mathbf{ U}_{{\mathcal S}{\mathcal F}}^*)^+]\det[\mathbf{ U}^*\mathbf{ T}^{k/2}]\\
=&\det[\mathbf{ T}^{k/2}\mathbf{ U}]\det[(\mathbf{ T}_{{\mathcal S}})^{-1}]\det[\mathbf{ U}^*\mathbf{ T}^{k/2}].\\
\end{split}
\end{equation}
The objective function to minimize the determinant of the error covariance matrix can be represented as
\begin{equation}
\mathcal{S}^*=\argmin_{\mathcal{S}\subset\mathcal{V}:|\mathcal{S}|=F}\det[\mathbf{ E}]
=\det[(\mathbf{ T}_{{\mathcal S}})^{-1}].
\end{equation}

\subsection{Proof of Proposition 1}
MaxCutoff also can be regarded as the SSS for minimizing the error caused by the reconstruction with ${\mathbf{T}}=(\mathbf{L}+\delta\mathbf{I})^{-1}$.
Because $A_T\sigma_\text{min}(\mathbf{T}_{\mathcal{S}})\leq\sigma_\text{min}((\mathbf{T}^{-1})_{\mathcal{S}^c})\leq B_T\sigma_\text{min}(\mathbf{T}_{\mathcal{S}})$ \cite{Govaer1989}, where $A_T$ and $B_T$ are some constant values determined from $\mathbf{T}$, \eqref{minspec} can be lower-bounded as
\begin{equation}
\begin{split}
\mathcal{S}^*&=\argmax_{\mathcal{S}\subset\mathcal{V}:|\mathcal{S}|=F}\sigma_\text{min}((\mathbf{ T}^{k})_{{\mathcal S}})\\
&\geq \argmax_{\mathcal{S}\subset\mathcal{V}:|\mathcal{S}|=F}\sigma_\text{min}((\mathbf{ T}^{-k})_{{\mathcal S}^c})\\
&= \argmax_{\mathcal{S}\subset\mathcal{V}:|\mathcal{S}|=F}\sigma_\text{min}(((\mathbf{L}+\delta\mathbf{I})^{k})_{{\mathcal S}^c}).\\
\end{split}
\label{maxcut_t}
\end{equation}
When $\delta$ goes to zero, \eqref{maxcut_t} becomes
\begin{equation}
\begin{split}
&\lim_{\delta\rightarrow 0}\argmax_{\mathcal{S}\subset\mathcal{V}:|\mathcal{S}|=F}\sigma_\text{min}(((\mathbf{L}+\delta\mathbf{I})^{k})_{{\mathcal S}^c})\\
&=\argmax_{\mathcal{S}\subset\mathcal{V}:|\mathcal{S}|=F}\sigma_\text{min}((\mathbf{ L}^k)_{\mathcal{S}^c}).\\
\end{split}
\label{maxcut_t2}
\end{equation}
As a result, maximizing the objective function of MaxCutoff shown in \eqref{gs1} leads to the minimization of the spectral norm of the error covariance matrix, where the error is caused by our reconstruction method in \eqref{eq3}.



\end{document}